\documentclass[11pt]{article}

\usepackage{amsmath,amsthm,amssymb}
\usepackage{fullpage}
\usepackage{hyperref}
\usepackage[usenames,dvipsnames,svgnames,table]{xcolor}
\usepackage{graphicx}
\usepackage{caption}
\usepackage{subcaption}
\usepackage{authblk}
\usepackage{algorithm}
\usepackage[noend]{algorithmic}

\usepackage{tikz}
\usetikzlibrary{arrows}
\usetikzlibrary{positioning, fit}
\usetikzlibrary{shapes} 
\usetikzlibrary{decorations.pathmorphing}
\usetikzlibrary{calc,through,intersections}

\newtheorem{theorem}{Theorem}
\newtheorem{problem}{Problem}
\newtheorem{corollary}[theorem]{Corollary}
\newtheorem{lemma}{Lemma}
\newtheorem{proposition}{Proposition}
\newtheorem{claim}{Claim}
\newtheorem{fact}{Fact}

\theoremstyle{definition}
\newtheorem{definition}{Definition}
\newtheorem{operation}{Operation}
\newtheorem{reduction}{Reduction rule}
\newtheorem{remark}{Remark}
\newtheorem{property}{Property}

\newsavebox{\mybox}

\newcommand{\qedclaim}{\hfill $\diamond$ \medskip}
\newenvironment{proofclaim}{\noindent{\em Proof.}}{\qedclaim}
\newenvironment{proofof}[1]{\medskip\noindent\emph{Proof of #1. }\ignorespaces}{\hfill\qed\medskip\par\noindent\ignorespacesafterend}

\title{The use of a pruned modular decomposition for {\sc Maximum Matching} algorithms on some graph classes
\thanks{This work was supported by the Institutional research programme PN 1819 "Advanced IT resources to support digital transformation processes in the economy and society - RESINFO-TD" (2018), project PN 1819-01-01"Modeling, simulation, optimization of complex systems and decision support in new areas of IT\&C research", funded by the Ministry of Research and Innovation, Romania.}}

\author[1,2]{Guillaume Ducoffe}
\author[1,3]{Alexandru Popa}

\affil[1]{\small National Institute for Research and Development in Informatics, Romania}
\affil[2]{\small The Research Institute of the University of Bucharest ICUB, Romania}
\affil[3]{\small University of Bucharest, Faculty of Mathematics and Computer Science}

\date{}

\begin{document}

\maketitle

\begin{abstract}
We address the following general question: given a graph class ${\cal C}$ on which we can solve {\sc Maximum Matching} in (quasi) linear time, does the same hold true for the class of graphs that can be {\em modularly decomposed} into ${\cal C}$ ?
A major difficulty in this task is that the {\sc Maximum Matching} problem is {\em not} preserved by quotient, thereby making difficult to exploit the structural properties of the quotient subgraphs of the modular decomposition.
So far, we are only aware of a recent framework in [Coudert et al., SODA'18] that only applies when the quotient subgraphs have bounded order and/or under additional assumptions on the nontrivial modules in the graph.
As a first attempt toward improving this framework we study the combined effect of modular decomposition with a pruning process over the quotient subgraphs.
More precisely, we remove sequentially from all such subgraphs their so-called one-vertex extensions ({\it i.e.}, pendant, anti-pendant, twin, universal and isolated vertices).
Doing so, we obtain a ``pruned modular decomposition'', that can be computed in ${\cal O}(m \log n)$-time.
Our main result is that if all the pruned quotient subgraphs have bounded order then a maximum matching can be computed in linear time.
This result is mostly based on two pruning rules on pendant and anti-pendant {\em modules} -- that are adjacent, respectively, to one or all but one other modules in the graph.
Furthermore, these two latter rules are surprisingly intricate and we consider them as our main technical contribution in the paper.

We stress that the class of graphs that can be totally decomposed by the pruned modular decomposition contains all the distance-hereditary graphs, and so, it is larger than cographs.
In particular, as a byproduct of our approach we also obtain the first known linear-time algorithms for {\sc Maximum Matching} on distance-hereditary graphs and graphs with modular-treewidth at most one.
Finally, we can use an extended version of our framework in order to compute a maximum matching, in linear-time, for all graph classes that can be modularly decomposed into cycles.
Our work is the first to explain why the existence of some nice ordering over the {\em modules} of a graph, instead of just over its vertices, can help to speed up the computation of maximum matchings on some graph classes.  
\end{abstract}


\section{Introduction}\label{sec:intro}

{\em Can we compute a maximum matching in a graph in linear-time?}
-- {\it i.e.}, computing a maximum set of pairwise disjoint edges in a graph. --
On one hand, despite considerable years of research and the design of elegant combinatorial and linear programming techniques, the best-known algorithms for this fundamental problem have stayed blocked to an ${\cal O}(m\sqrt{n})$-time complexity~\cite{Gab17,MiV80}.
On the other hand, the {\sc Maximum Matching} problem has several applications~\cite{Bun00,DeS84,LoP09,Pul95}, some of them being relevant only for specific graph families.
For instance, there exists a special matching problem arising in industry that can be rephrased as finding a maximum matching in a given convex bipartite graph~\cite{Glo67}.
It may then be tempting to use some well-structured graph classes in order to overcome this superlinear barrier for particular cases of graphs.
We follow this well-established line of research ({\it e.g.}, see~\cite{Cha96,CDP18,DaK98,Dra97,FPT97,FGV99,Glo67,HoK73,KaS81,LiR93,MNN17,MoJ89,YuY93,YuZ07,Yus13}).
Our work will combine two successful approaches for this problem, namely, the use of a {\em vertex-ordering} characterization for certain graph classes, and a recent technique based on the decomposition of a graph by its {\em modules}.
We detail these two approaches in what follows, before summarizing our contributions.  

\subsection{Related work}

A cornerstone of most {\sc Maximum Matching} algorithms is the notion of augmenting paths~\cite{Ber57,Edm65}.
-- See~\cite{FLPS+17,Lov79,MVV87} for some other approaches based on matrix multiplication. --
Roughly, given some matching $F$ in a graph, an $F$-augmenting path is a simple path between two exposed vertices ({\it i.e.}, not part of $V(F)$) in which the edges belong alternatively to $F$ and not to $F$.
The symmetric difference between $F$ and any $F$-augmenting path leads to a new matching of greater cardinality $|F|+1$.
Therefore, the standard strategy in order to solve {\sc Maximum Matching} is to repeatedly compute a set of vertex-disjoint augmenting paths (w.r.t. the current matching) until no more such path can be found.
However, although we can compute a set of augmenting paths in linear-time~\cite{GaT83}, this is a tedious task that involves the technical notion of blossoms and this may need to be repeated $\Omega(\sqrt{n})$ times before a maximum matching can be computed~\cite{HoK73}.
One way to circumvent this issue for a specific graph class is to use the structural properties of this class in order to compute the augmenting paths more efficiently.

For instance, a well-known greedy approach consists in, given some total ordering $(v_1,v_2,\ldots,v_n)$ over the vertices in the graph, to consider the exposed vertices $v_i$ by increasing order, then to try to match them with some exposed neighbour $v_j$ that appears later in the ordering~\cite{Dra97}.
The candidate vertex $v_j$ can be chosen either arbitrarily or according to some specific rules depending on the graph class we consider.
On the positive side, variants of the greedy approach have been used in order to compute maximum matchings in quasi linear-time on some graph classes that admit a vertex-ordering characterization, such as: interval graphs~\cite{MoJ89}, convex bipartite graphs~\cite{Glo67}, strongly chordal graphs~\cite{DaK98}, cocomparability graphs~\cite{MNN17}, etc.
But on the negative side, the greedy approach does not look promising for chordal graphs (defined by the existence of a perfect elimination ordering) since computing a maximum matching in a given split graph is already ${\cal O}(n^2)$-time equivalent to the same problem on general bipartite graphs~\cite{DaK98}.

More related to our work are the reduction rules, presented in~\cite{KaS81}, for the vertices of degree at most two.
Indeed, since a pendant vertex in any path, and in particular in an augmenting one, can only be an endpoint, it can be easily seen that we can always match a pendant vertex with its unique neighbour w.l.o.g.
A slightly more complicated reduction rule is presented for degree-two vertices in~\cite{KaS81}.
Nevertheless, no such rule is likely to exist already for degree-three vertices since {\sc Maximum Matching} on cubic graphs is linear-time equivalent to {\sc Maximum Matching} on general graphs~\cite{Bie01}.
Our initial goal was to extend similar reduction rules to {\em module-orderings} -- defined next -- of which vertex-orderings are a particular case.

\paragraph{\sc Modular Decomposition.}
A module in a graph $G=(V,E)$ is any vertex-subset $X$ such that every vertex of $V \setminus X$ is either adjacent to every of $X$ or nonadjacent to every of $X$.
Trivial examples of modules are $\emptyset, V$ and $\{v\}$ for every $v \in V$.
Roughly, the {\em modular decomposition} of $G$ is a recursive decomposition of $G$ according to its nontrivial modules~\cite{HaP10}.
We postpone its formal definition until Section~\ref{sec:prelim}.
For now, we only want to stress that the vertices in the ``quotient subgraphs'' that are outputted by this decomposition represent modules of $G$.

The use of modular decomposition in the algorithmic field has a rich history.
The successive improvements on the best-known complexity for computing this decomposition are already interesting on their own since they required the introduction of several new techniques~\cite{CoH94,DGM01,HMP04,MRS94,TCHP08}.
There is now a practical linear-time algorithm for computing the modular decomposition of any graph~\cite{TCHP08}.
Furthermore, the story does not end here since modular decomposition was also applied in order to solve several graph-theoretic problems ({\it e.g.}, see~\cite{ALMH+17,BKM06,CDP18,Dah98,FLMT18,GLO13}).
Our main motivation for considering modular decomposition in this note is its recent use in the field of parameterized complexity for {\em polynomial} problems.
-- For some earlier applications to NP-hard problems, see~\cite{GLO13}. --
More precisely, let us call {\em modular-width} of a graph $G$ the minimum $k \geq 2$ such that every quotient subgraph in the modular decomposition of $G$ is either ``degenerate'' ({\it i.e.}, complete or edgeless) or of order at most $k$.
With Coudert, we proved in~\cite{CDP18} that many ``hard'' graph problems in P -- for which no linear-time algorithm is likely to exist -- can be solved in $k^{{\cal O}(1)}(n+m)$-time on graphs with modular-width at most $k$.
In particular, we proposed an ${\cal O}(k^4n +m)$-time algorithm for {\sc Maximum Matching}.

One appealing aspect of our approach in~\cite{CDP18} was that, for most problems studied, we obtained a linear-time reduction from the input graph $G$ to some (smaller) quotient subgraph $G'$ in its modular decomposition.
-- We say that the problem is preserved by quotient. --
This paved the way to the design of efficient algorithms for these problems on graph classes with {\em unbounded} modular-width, assuming their quotient subgraphs are simple enough w.r.t. the problem at hands.
We illustrated this possibility through the case of $(q,q-3)$-graphs ({\it i.e.}, graphs where no set of at most $q$ vertices, $q \geq 7$, can induce more than $q-3$ paths of length four).
However, this approach completely fell down for {\sc Maximum Matching}.
Indeed, our {\sc Maximum Matching} algorithm in~\cite{CDP18} works on supergraphs of the quotient graphs that need to be repeatedly updated every time a new augmenting path is computed.
Such approach did not help much in exploiting the structure of quotient graphs.
We managed to do so for $(q,q-3)$-graphs only through the help of a deeper structural theorem on the nontrivial modules in this class of graphs.
Nevertheless, to take a shameful example, it was not even known before this work whether {\sc Maximum Matching} could be solved faster than with the state-of-the art algorithms on the graphs that can be modularly decomposed into {\em paths}!

\subsection{Our contributions}

We propose {\em pruning rules} on the modules in a graph (some of them new and some others revisited) that can be used in order to compute {\sc Maximum Matching} in linear-time on several new graph classes.
More precisely, given a module $M$ in a graph $G=(V,E)$, recall that $M$ is corresponding to some vertex $v_M$ in a quotient graph $G'$ of the modular decomposition of $G$.
Assuming $v_M$ is a so-called {\em one-vertex extension} in $G'$ ({\it i.e.}, it is pendant, anti-pendant, universal, isolated or it has a twin), we show that a maximum matching for $G$ can be computed from a maximum matching of $G[M]$ and a maximum matching of $G \setminus M$ efficiently (see Section~\ref{sec:core}).
Our rules are purely {\em structural}, in the sense that they only rely on the structural properties of $v_M$ in $G'$ and not on any additional assumption on the nontrivial modules.
Some of these rules ({\it e.g.}, for isolated or universal modules) were first introduced in~\cite{CDP18} --- although with slightly different correctness proofs.
Our main technical contributions in this work are the pruning rules for, respectively, {\em pendant} and {\em anti-pendant} modules (see Sections~\ref{sec:anti-pendant} and~\ref{sec:pendant}). 
The latter two cases are surprisingly the most intricate.
In particular, they require amongst other techniques: the computation of specified augmenting paths of length up to $7$, the addition of some ``virtual edges'' in other modules, and a careful swapping between some matched and unmatched edges. 

Then, we are left with pruning every quotient subgraph in the modular decomposition by sequentially removing the one-vertex extensions.
We prove that the resulting ``pruned quotient subgraphs'' are unique (independent from the removal orderings) and that they can be computed in quasi linear-time using a {\em trie} data-structure (Section~\ref{sec:oruned-mod-dec}).
Furthermore, many interesting graph classes are totally decomposable w.r.t. this new ``pruned modular decomposition'', namely, every graph that can be decomposed into: trees, distance-hereditary graphs~\cite{BaM86}, tree-perfect graphs~\cite{BrL99}, etc. 
These classes are further discussed in Section~\ref{sec:applications}.
Note that for some of them, such as distance-hereditary graphs, we so obtain the first known linear-time algorithm for {\sc Maximum Matching}.
With slightly more work, we can extend our approach to every graph that can be modularly decomposed into cycles (Section~\ref{sec:unicycle}).
The case of graphs that can be modularly decomposed into {\em series-parallel graphs}~\cite{Epp92}, or more generally the graphs of bounded {\em modular treewidth}~\cite{PSS16}, is left as an interesting open question.

\bigskip
Definitions and our first results are presented in Section~\ref{sec:prelim}. 
We introduce the pruned modular decomposition in Section~\ref{sec:oruned-mod-dec}, where we show that it can be computed in quasi linear-time.
Then, the core of the paper is Section~\ref{sec:core} where the pruning rules are presented along with their correctness proofs. 
In particular, we state our main result in Section~\ref{sec:main}.
Applications of our approach to some graph classes are discussed in Section~\ref{sec:applications}.
Finally, we conclude in Section~\ref{sec:ccl} with some open questions.

\section{Preliminaries}\label{sec:prelim}

For the standard graph terminology, see~\cite{BoM08,Die10}.
We only consider graphs that are finite, simple (hence without loops or multiple edges), and unweighted -- unless stated otherwise.
Furthermore we make the standard assumption that graphs are encoded as adjacency lists.
In what follows, we introduce our main algorithmic tool for the paper as well as the graph problems we study.

\subsection*{Modular decomposition}

We define a {\em module} in a graph $G=(V,E)$ as any subset $M \subseteq V(G)$ such that for any $u,v \in M$ we have $N_G(v) \setminus M = N_G(u) \setminus M$.
Said otherwise, every vertex of $V \setminus M$ is either adjacent to every vertex of $M$ or nonadjacent to every vertex of $M$.
There are trivial examples of modules such as $\emptyset, \ V, \ \mbox{and} \ \{v\}$ for every $v \in V$.
Other less trivial examples of modules are the connected components of $G$ and, similarly, the co-connected components of $G$.
Let ${\cal P} = \{M_1,M_2,\ldots,M_p\}$ be a partition of the vertex-set $V$.
If for every $1 \leq i \leq p$, $M_i$ is a module of $G$, then we call ${\cal P}$ a {\em modular partition} of $G$.
By abuse of notation, we will sometimes identify a module $M_i$ with the induced subgraph $H_i = G[M_i]$, {\it i.e.}, we will write ${\cal P} = \{ H_1, H_2, \ldots H_p\}$.

The {\em quotient subgraph} $G / {\cal P}$ has for vertex-set ${\cal P}$, and there is an edge between every two modules $M_i,M_j \in {\cal P}$ such that $M_i \times M_j \subseteq E$.
Conversely, let $G'=(V',E')$ be a graph and let ${\cal P} = \{ H_1, H_2, \ldots H_p\}$. be a collection of subgraphs.
The {\em substitution graph} $G'({\cal P})$ is obtained from $G'$ by replacing every vertex $v_i \in V'$ with a module inducing $H_i$.
In particular, for $G' = G / {\cal P}$ we have that $G'({\cal P}) = G$.

\medskip
The {\em modular decomposition} of $G$ is a compact representation of all the modules in $G$ (in ${\cal O}(n+m)$-space), that can be recursively defined as follows.
First we say that $G$ is {\em prime} if its only modules are trivial ({\it i.e.}, $\emptyset, \ V, \ \mbox{and the singletons} \ \{v\}$).
Roughly, the modular decomposition of $G$ is a tree-like collection of prime quotient subgraphs of $G$. 
In order to make this more precise, let us introduce a few more notions.
We call a module $M$ {\em strong} if it does not overlap any other module, {\it i.e.}, for any module $M'$ of $G$, either one of $M$ or $M'$ is contained in the other or $M$ and $M'$ do not intersect. 
Let ${\cal M}(G)$ be the family of all inclusion wise maximal strong modules of $G$ that are proper subsets of $V$.
The family ${\cal M}(G)$ is a modular partition of $G$~\cite{HaP10}, and so, we can define $G' = G / {\cal M}(G)$.
The following structure theorem is due to Gallai.

\begin{theorem}[~\cite{Gal67}]\label{thm:modular-dec}
For an arbitrary graph $G$ exactly one of the following conditions is satisfied.
\begin{enumerate}
\item $G$ is disconnected;
\item its complement $\overline{G}$ is disconnected;
\item or its quotient graph $G'$ is prime for modular decomposition.
\end{enumerate}
\end{theorem}

Then, the modular decomposition of $G$ is formally defined as follows.
We output the quotient graph $G'$ and, for any strong module $M \in {\cal M}(G)$ that is nontrivial (possibly none if $G=G'$), we also output the modular decomposition of $G[M]$.
Observe that, by Theorem~\ref{thm:modular-dec}, the subgraphs from the modular decomposition are either edgeless, complete, or prime for modular decomposition. 
See Fig.~\ref{fig:modular-dec} for an example.

\begin{figure}[h!]
\centering
\begin{subfigure}[b]{.46\textwidth}\centering
\resizebox{!}{4cm}{ \begin{tikzpicture}

\node[circle,fill=black,inner sep=0pt,minimum size=5pt,label=above:{$a$}] (a) at (0.8, 4.5) {};
\node[circle,fill=black,inner sep=0pt,minimum size=5pt,label=above:{$b$}] (b) at (2.2, 4.5) {};

\node[circle,fill=black,inner sep=0pt,minimum size=5pt,label=above left:{$c$}] (c) at (0, 3) {};
\node[circle,fill=black,inner sep=0pt,minimum size=5pt,label=below left:{$d$}] (d) at (1, 3) {};
\node[circle,fill=black,inner sep=0pt,minimum size=5pt,label=below right:{$e$}] (e) at (2, 3) {};
\node[circle,fill=black,inner sep=0pt,minimum size=5pt,label=above right:{$f$}]  (f) at (3, 3) {};

\node[circle,fill=black,inner sep=0pt,minimum size=5pt,label=right:{$g$}] (g) at (1.5, 1.5) {};

\node[circle,fill=black,inner sep=0pt,minimum size=5pt,label=below left:{$h$}] (h) at (0, 0) {};
\node[circle,fill=black,inner sep=0pt,minimum size=5pt,label=above left:{$i$}] (i) at (1, 0) {};
\node[circle,fill=black,inner sep=0pt,minimum size=5pt,label=above right:{$j$}] (j) at (2, 0) {};
\node[circle,fill=black,inner sep=0pt,minimum size=5pt,label=below right:{$k$}]  (k) at (3, 0) {};

\draw (a) -- (c) -- (b) -- (d) -- (a) -- (e) -- (b) -- (f) -- (a);
\draw (c) -- (d) -- (e) -- (f);
\draw (c) -- (g) -- (d);
\draw (e) -- (g) -- (f);
\draw (h) -- (g) -- (i);
\draw (j) -- (g) -- (k);
\draw (h) -- (i) -- (j) -- (k);
\draw (h) to [out=-30,in=-150] (k);
\end{tikzpicture} }
\end{subfigure}\hfill
\begin{subfigure}[b]{.52\textwidth}\centering
\resizebox{!}{4.5cm}{ \begin{tikzpicture}
\node (ab) [rectangle,fill=black,minimum size=8pt,label=above:{$ab$}, anchor=center] {};
\node (cdef) [rectangle,fill=black,minimum size=8pt, label=above:{$cdef$}, anchor=center, right = 1cm of ab] {};
\node (g) [circle,fill=black,inner sep=0pt,minimum size=8pt,label=above:{$g$}, right = 1cm of cdef]  {};
\node (hijk) [rectangle,fill=black,minimum size=8pt, label=above:{$hijk$}, anchor=center, right = 1cm of g] {};
\draw (ab) -- (cdef) -- (g) -- (hijk);
\node (block1) [style={ellipse, x radius=5cm, y radius=3cm, draw=black, inner sep=3ex}, label=right:{prime}, fit={(ab) (cdef) (g) (hijk)}] {};

\node (a) [circle,fill=black,inner sep=0pt,minimum size=8pt,label=left:{$a$}, below= 1cm of block1, xshift=-5cm]  {};
\node (b) [circle,fill=black,inner sep=0pt,minimum size=8pt,label=left:{$b$}, right = .5cm of a, yshift=3mm]  {};
\node (block2) [style={ellipse, x radius=2cm, y radius=1cm, draw=black, inner sep=1ex}, label=below:{edgeless}, fit={(a) (b)}] {};

\node (c) [circle,fill=black,inner sep=0pt,minimum size=8pt,label=right:{$c$}, below= 1.5cm of block1, xshift=-1.5cm]  {};
\node (d) [circle,fill=black,inner sep=0pt,minimum size=8pt,label=left:{$d$}, below= .5cm of c, xshift=3mm]  {};
\node (e) [circle,fill=black,inner sep=0pt,minimum size=8pt,label=left:{$e$}, below= .5cm of d, xshift=3mm]  {};
\node (f) [circle,fill=black,inner sep=0pt,minimum size=8pt,label=left:{$f$}, below= .5cm of e, xshift=3mm]  {};
\draw (c) -- (d) -- (e) -- (f);
\node (block3) [style={ellipse, x radius=2cm, y radius=1cm, draw=black, inner sep=1ex}, label=below:{prime}, fit={(c) (d) (e) (f)}] {};

\node (hj) [rectangle,fill=black,minimum size=8pt,label=above:{$hj$}, below= 1.5cm of block1, xshift=3cm]  {};
\node (ik) [rectangle,fill=black,minimum size=8pt,label=above:{$ik$}, right = .5cm of hj, yshift=-3mm]  {};
\draw (hj) -- (ik);
\node (block4) [style={ellipse, x radius=2cm, y radius=1cm, draw=black, inner sep=2ex}, label=above:{complete}, fit={(hj) (ik)}] {};

\node (j) [circle,fill=black,inner sep=0pt,minimum size=8pt,label=left:{$j$}, below= 1cm of block4, xshift=-1cm]  {};
\node (h) [circle,fill=black,inner sep=0pt,minimum size=8pt,label=left:{$h$},  left = .5cm of j, yshift=-3mm]  {};
\node (block5) [style={ellipse, x radius=2cm, y radius=1cm, draw=black, inner sep=1ex}, label=below:{edgeless}, fit={(j) (h)}] {};

\node (i) [circle,fill=black,inner sep=0pt,minimum size=8pt,label=left:{$i$}, below = 1cm of block4, xshift=1cm]  {};
\node (k) [circle,fill=black,inner sep=0pt,minimum size=8pt,label=left:{$k$},  right = .5cm of i, yshift=-3mm]  {};
\node (block6) [style={ellipse, x radius=2cm, y radius=1cm, draw=black, inner sep=1ex}, label=below:{edgeless}, fit={(i) (k)}] {};

\draw (block1) -- (block2);
\draw (block3) -- (block1);
\draw (block4) -- (block1);
\draw (block4) -- (block5);
\draw (block4) -- (block6);
\end{tikzpicture} }
\end{subfigure}
\caption{A graph and its modular decomposition.}
\label{fig:modular-dec}
\end{figure}
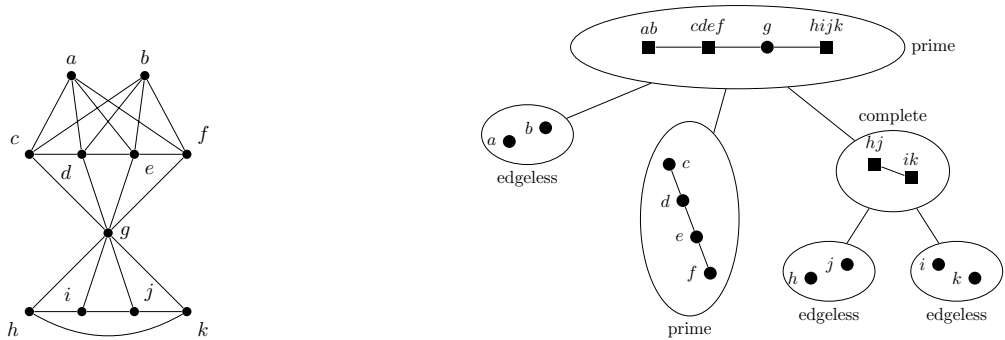

The modular decomposition of a given graph $G=(V,E)$ can be computed in linear-time~\cite{TCHP08}.
We stress that there are many graph classes that can be characterized using the modular decomposition~\cite{ABNT16}.
In particular, $G$ is a cograph if and only if every quotient subgraph in its modular decomposition is either complete or disconnected~\cite{CPS85}.

\subsection*{Maximum Matching}

A matching in a graph is defined as a set of edges with pairwise disjoint end vertices.
The maximum cardinality of a matching in a given graph $G=(V,E)$ is denoted by $\mu(G)$.
We consider the problem of computing a matching of maximum cardinality.

\begin{center}
	\fbox{
		\begin{minipage}{.95\linewidth}
			\begin{problem}[\textsc{Maximum Matching}]\
				\label{prob:max-matching} 
					\begin{description}
					\item[Input:] A graph $G=(V,E)$.
					\item[Output:] A matching of $G$ with maximum cardinality.
				\end{description}
			\end{problem}     
		\end{minipage}
	}
\end{center}

We remind the reader that {\sc Maximum Matching} can be solved in ${\cal O}(m\sqrt{n})$-time on general graphs --- although we do not use this result directly in our paper~\cite{Gab17,MiV80}.
Furthermore, let $G=(V,E)$ be a graph and let $F \subseteq E$ be a matching of $G$.
We call a vertex matched if it is incident to an edge of $F$, and exposed otherwise.
Then, we define an $F$-augmenting path as a path where the two ends are exposed, and the edges belong alternatively to $F$ and not to $F$.
It is well-known and easy to check that, given an $F$-augmenting path $P = (v_1,v_2, \ldots, v_{2\ell})$, the matching $E(P)\Delta F$ (obtained by symmetric difference on the edges) has larger cardinality than $F$.

\begin{lemma}[Berge,~\cite{Ber57}]\label{lem:berge}
A matching $F$ in $G=(V,E)$ is maximum if and only if there is no $F$-augmenting path.
\end{lemma}

In this paper, we will consider an intermediate matching problem, first introduced in~\cite{CDP18}.

\begin{center}
	\fbox{
		\begin{minipage}{.95\linewidth}
			\begin{problem}[\textsc{Module Matching}]\
				\label{prob:mod-matching} 
					\begin{description}
					\item[Input:] A graph $G'=(V',E')$ with the following additional information;
                                  \begin{itemize}
                                    \item a collection of subgraphs ${\cal P} = \{ H_1, H_2, \ldots, H_p \}$;

                                    \item a collection ${\cal F} = \{ F_1, F_2, \ldots, F_p \}$, 

                                    with $F_i$ being a maximum matching of $H_i$ for every $i$.
                                  \end{itemize}
					\item[Output:] A matching of $G = G'({\cal P})$ with maximum cardinality.
				\end{description}
			\end{problem}     
		\end{minipage}
	}
\end{center}

A natural choice for {\sc Module Matching} would be to take ${\cal P} = {\cal M}(G)$.
However, we will allow ${\cal P}$ to take different values for our reduction rules.
In particular, by setting $G' = G$, ${\cal P} = V$ and ${\cal F}$ a collection of empty matchings, one obtains that \textsc{Module Matching} is a strict generalization of {\sc Maximum Matching}.

\paragraph{Additional notations.}
Let $\langle G',{\cal P},{\cal F} \rangle$ be any instance of \textsc{Module Matching}.
The order of $G'$, equivalently the cardinality of ${\cal P}$, is denoted by $p$.
For every $1 \leq i \leq p$ let $M_i = V(H_i)$ and let $n_i = |M_i|$ be the order of $H_i$.
We denote $\delta_i = |E(M_i,\overline{M_i})|$ the size of the cut with all the edges between $M_i$ and $N_G(M_i)$.
In particular, we have $\delta_i = \sum_{v_j \in N_{G'}(v_i)} n_in_j$. 
Let us define $\Delta m(G') = \sum_i \delta_i$.
We will omit the dependency in $G'$ if it is clear from the context.
Finally, let $\Delta \mu = \mu(G) - \sum_i \mu(H_i)$. 

\medskip
Our framework is based on the following lemma (inspired from~\cite{CDP18}).

\begin{lemma}\label{lem:incr}
Let $G=(V,E)$ be a graph.
Suppose that for every $H'$ in the modular decomposition of $G$ we can solve \textsc{Module Matching} on any instance $\langle H',{\cal P},{\cal F} \rangle$ in time $T(p,\Delta m,\Delta \mu)$, where $T$ is a subadditive function\footnote{We stress that every polynomial function is subadditive.}.
Then, we can solve \textsc{Maximum Matching} on $G$ in time ${\cal O}(T({\cal O}(n),m,n))$.
\end{lemma}

\begin{proof}
Let $N$ be the sum of the orders of all the subgraphs in the modular decomposition of $G$.
Next, we describe an algorithm for \textsc{Maximum Matching} that runs in time ${\cal O}(T(N,m,\mu(G)))$.  
The latter will prove the lemma since we have $N = {\cal O}(n)$~\cite{HaP10}.
We prove our result by induction on the number of subgraphs in the modular decomposition of $G$.
For that, let $G'=({\cal M}(G),E')$ be the quotient graph of $G$.
There are two cases.

Suppose $G = G'$, or equivalently the modular decomposition of $G$ is reduced to $G'$.
Then, solving {\sc Maximum Matching} for $G$ is equivalent to solving {\sc Module Matching} on $\langle G,{\cal M}(G), {\cal F} \rangle$ with ${\cal M}(G) = V$ and ${\cal F}$ being a collection of $n$ empty matchings.
By the hypothesis it can be done in time $T(p,\Delta m,\Delta \mu)$, where $p = |V(G')|$.
In this situation we get $p = |V| = N$, $\Delta m = |E| = m$, $\Delta \mu = \mu(G)$.
Since $T$ is subadditive by the hypothesis, we have that {\sc Maximum Matching} can be solved in this case in time ${\cal O}(T(N,m,\mu(G)))$. 

Otherwise, $G \neq G'$.
Let ${\cal M}(G) = \{M_1, M_2, \ldots, M_p\}$.
For every $1 \leq i \leq p$, we call the algorithm recursively on $H_i = G[M_i]$ and we so obtain a maximum matching $F_i$ for this subgraph.
By the induction hypothesis, this step takes times ${\cal O}(\sum_{i=1}^p T(N_i,m_i,\mu(H_i)))$, with $N_i$ being the sum of the orders of all the subgraphs in the modular decomposition of $H_i$ and $m_i = |E(H_i)|$.
Furthermore, let ${\cal F} = \{F_1,F_2,\ldots,F_p\}$.
Observe that we have $\sum_{i=1}^p N_i = N - p, \ \sum_{i=1}^p m_i = m - \Delta m \ \mbox{and} \ \sum_{i=1}^p \mu(H_i) = \mu(G) - \Delta \mu$.
In order to compute a maximum matching for $G$, we are left with solving {\sc Module Matching} on $\langle G', {\cal M}(G), {\cal F} \rangle$, that takes time $T(p,\Delta m,\Delta \mu)$ by the hypothesis.
Overall, since $T$ is subadditive, the total running time is an ${\cal O}(T(N,m,\mu(G)))$. 
\end{proof}

An important observation for our subsequent analysis is that, given any module $M$ of a graph $G$, the internal structure of $G[M]$ has no more relevance after we computed a maximum matching $F_M$ for this subgraph.
More precisely, we will use the following lemma:

\begin{lemma}[~\cite{CDP18}]\label{lem:mw-matching-reduction}
Let $M$ be a module of $G=(V,E)$, let $G[M] = (M,E_M)$ and let $F_M  \subseteq E_M$ be a maximum matching of $G[M]$.
Then, every maximum matching of $G_M' = (V, (E \setminus E_M) \cup F_M)$ is a maximum matching of $G$.
\end{lemma}

By Lemma~\ref{lem:mw-matching-reduction} we can modify our algorithmic framework as follows.
For every instance $\langle G', {\cal P}, {\cal F}\rangle$ for {\sc Module Matching}, we can assume that $H_i = (M_i,F_i)$ for every $1 \leq i \leq p$.

\paragraph{Data structures.}
Finally, let $\langle G', {\cal P}, {\cal F}\rangle$ be any instance for {\sc Module Matching}.
A {\em canonical ordering} of $H_i$ (w.r.t. $F_i)$ is a total ordering over $V(H_i)$ such that the exposed vertices appear first, and every two vertices that are matched together are consecutive.
In what follows, we will assume that we have access to a canonical ordering for every $i$.
Such orderings can be computed in time ${\cal O}(\sum_i |M_i| + |F_i|)$ by scanning all the modules and the matchings in ${\cal F}$, that is an ${\cal O}(\Delta m)$ provided $G'$ is connected.

Furthermore, let $F$ be a (not necessarily maximum) matching for the subdivision $G = G'({\cal P})$.
We will make the standard assumption that, for every $v \in V(G)$, we can decide in constant-time whether $v$ is matched by $F$, and if so, we can also access in constant-time to the vertex matched with $v$.

\section{A pruned modular decomposition}\label{sec:oruned-mod-dec}

In this section, we introduce a pruning process over the quotient subgraphs, that we use in order to refine the modular decomposition.

\begin{definition}\label{def:one-vertex-ext}
Let $G=(V,E)$ be a graph.
We call $v \in V$ a one-vertex extension if it falls in one of the following cases:
\begin{itemize}
\item $N_G[v] = V$ ({\em universal}) or $N_G(v) = \emptyset$ ({\em isolated}); 
\item $N_G[v] = V \setminus u$ ({\em anti-pendant}) or $N_G(v) = \{u\}$ ({\em pendant}), for some $u \in V \setminus v$;
\item $N_G[v] = N_G[u]$ ({\em true twin}) or $N_G(v) = N_G(u)$ ({\em false twin}), for some $u \in V \setminus v$.
\end{itemize}
\end{definition}

A pruned subgraph of $G$ is obtained from $G$ by sequentially removing one-vertex extensions (in the current subgraph) until it can no more be done.
This terminology was introduced in~\cite{LRT00}, where they only considered the removals of twin and pendant vertices.
Also, the clique-width of graphs that are totally decomposed by the above pruning process ({\it i.e.}, with their pruned subgraph being a singleton) was studied in~\cite{Rao08}
\footnote{Anti-twins are also defined as one-vertex extensions in~\cite{Rao08}. Their integration to this framework remains to be done.}.
Our contribution in this part is twofold.
First, we show that the gotten subgraph is ``almost'' independent of the removal ordering, {\it i.e.}, there is a unique pruned subgraph of $G$ (up to isomorphism).
The latter can be derived from the following (easy) lemma:

\begin{lemma}\label{lem:pruned-unique}
Let $G=(V,E)$ be a graph and let $v,v' \in V$ be one-vertex extensions of $G$.
Suppose that $v$ and $v'$ are {\em not} pairwise twins.
Then, $v'$ is also a one-vertex extension of $G \setminus v$.
\end{lemma}

\begin{proof}
We need to consider several cases.
If $v'$ is either isolated or universal in $G$ then it stays so in $G \setminus v$.
If $v'$ is pendant in $G$ then it is either pendant or isolated in $G \setminus v$.
Similarly, if $v'$ is anti-pendant in $G$ then it is either anti-pendant or universal in $G \setminus v$.
Otherwise, $v'$ has a twin $u$ in $G$.
By the hypothesis, $u \neq v'$.
Then, we have that $u,v$ stay pairwise twins in $G \setminus v'$.
\end{proof}

\begin{corollary}\label{cor:unique-pruned}
Every graph $G=(V,E)$ has a unique pruned subgraph up to isomorphism.
\end{corollary}

\begin{proof}
Suppose for the sake of contradiction that $G$ has two non-isomorphic pruned subgraphs.
W.l.o.g., $G$ is a minimum counter-example.
In particular, for every one-vertex extension $v$ of $G$, we have that $G \setminus v$ has a unique pruned subgraph up to isomorphism.
Therefore, there exist $v,v' \in V$ such that: $v,v'$ are one-vertex extensions of $G$, and the pruned subgraphs of $G \setminus v$ and $G \setminus v'$ are non isomorphic.
We claim that $v$ is {\em not} a one-vertex extension of $G \setminus v'$ (resp., $v'$ is not a one-vertex extension of $G \setminus v$).
Indeed, otherwise, both the pruned subgraphs of $G \setminus v$ and of $G \setminus v'$ would be isomorphic to the pruned subgraphs of $G \setminus \{v,v'\}$.
By Lemma~\ref{lem:pruned-unique}, it implies that $v,v'$ are pairwise twins in $G$.
However, since $G \setminus v$ and $G \setminus v'$ are isomorphic, so are their respective pruned subgraphs.
A contradiction.
\end{proof}

For most classes of graphs that we consider in this note, the pruned subgraph is trivial (reduced to a singleton) and can be computed in linear-time.
For purpose of completeness, we observe in what follows that the same can be done for any graph (up to a logarithmic factor).

\begin{proposition}\label{prop:pruned-compute}
For every $G=(V,E)$, its pruned subgraph can be computed in ${\cal O}(n+m\log n)$-time.
\end{proposition}

\begin{proof}
By Corollary~\ref{cor:unique-pruned}, we are left with greedily searching for, then eliminating, the one-vertex extensions. 
We can compute the ordered degree sequence of $G$ in ${\cal O}(n+m)$-time.
Furthermore, after any vertex $v$ is eliminated, we can update this sequence in ${\cal O}(|N(v)|)$-time.
Hence, up to a total update time in ${\cal O}(n+m)$, at any step we can detect and remove in constant-time any vertex that is either universal, isolated, pendant or anti-pendant.
Finally, in~\cite{LRT00} they proposed a trie data-structure supporting the following two operations: suppression of a vertex; and detection of true or false twins (if any).
The total time for all the operations on this data-structure is in ${\cal O}(n+m\log n)$~\cite{LRT00}.
\end{proof}

From now on, we will term ``pruned modular decomposition'' of a graph $G$ the collection of the pruned subgraphs for all the quotient subgraphs in the modular decomposition of $G$.
Note that there is a unique pruned modular decomposition of $G$ (up to isomorphism) and that it can be computed in ${\cal O}(n+m\log n)$-time by Proposition~\ref{prop:pruned-compute} (applied to every quotient subgraph in the modular decomposition separately).

\begin{remark}
A careful reader could object that most cases of one-vertex extensions (universal, isolated, twin) imply the existence of non trivial modules.
Therefore, such vertices should not exist in the prime quotient subgraphs of the modular decomposition.
However, they may appear after removal of pendant or anti-pendant vertices (see Fig.~\ref{fig:bull}).
\end{remark}

\begin{figure}[h!]
\centering
\includegraphics[width=.45\textwidth]{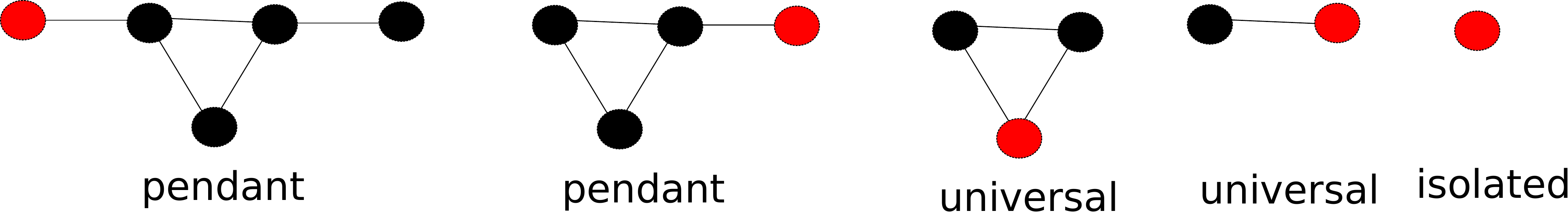}
\caption{The bull is totally decomposable by the pruned modular decomposition.}
\label{fig:bull}
\end{figure}

\section{Reduction rules}\label{sec:core}

In this section, we consider any instance $\langle G',{\cal P},{\cal F} \rangle$ of {\sc Module Matching}.
We suppose that $v_1$, the vertex corresponding to $M_1$ in $G'$ is a one-vertex extension
Under this assumption, we present reduction rules to a smaller instance $\langle G^*,{\cal P}^*,{\cal F}^* \rangle$ where $|{\cal P}^*| < |{\cal P}|$.
Some of the rules we propose next were first introduced in~\cite{CDP18}, namely for universal and isolated modules (explicitly) and for false or true twin modules (implicitly).
We state these above rules in Section~\ref{sec:simple} for completeness of the paper.
Our main technical contributions are the reduction rules for pendant and anti-pendant modules (presented in Sections~\ref{sec:anti-pendant} and~\ref{sec:pendant}, respectively), which are surprisingly the most intricate.
Finally, we end this section stating our main result (Theorem~\ref{thm:main}).

\subsection{Simple cases}\label{sec:simple}

We start with two cases that can be readily solved, namely:

\begin{reduction}[see also~\cite{CDP18}]\label{rule:isolated}
Suppose $v_1$ is isolated in $G'$.

We set $G^* = G' \setminus v_1$, ${\cal P}^* = {\cal P} \setminus \{H_1\}$, and ${\cal F}^* = {\cal F} \setminus \{F_1\}$.
\end{reduction}

Indeed, note that in this above case a maximum matching of $G$ is the union of $F_1$ with any maximum matching of the subdivision $G^*({\cal H}^*)$.

\begin{reduction}\label{rule:false-twin}
Suppose $v_1,v_2$ are false twins in $G'$.

We set $G^* = G' \setminus v_1$, ${\cal P}^* = \{H_1 \cup H_2\} \cup ({\cal P} \setminus \{H_1,H_2\})$, ${\cal F}^* = \{F_1 \cup F_2\} \cup ({\cal F} \setminus \{F_1,F_2\})$.
\end{reduction}

Indeed, note that in this above case, $M_1 \cup M_2$ is a module of $G$.
Furthermore, since there is no edge between $M_1$ and $M_2$ we have that $F_1 \cup F_2$ is a maximum matching of $G[M_1 \cup M_2] = H_1 \cup H_2$.
Hence, the above reduction rule is correct.

\medskip
Then, before we state the two other reduction rules in this section, we need to introduce the following two local operations on a matching (they will be used for all the remaining reduction rules).
To our best knowledge, these two following operations were first introduced in~\cite{YuY93} for the computation of maximum matchings in cographs.
We give a sligthly generalized presentation of these rules.
In what follows, let $F \subseteq E$ be a matching and let $M \subseteq V$ be a module.

\begin{operation}[{\sc MATCH}$(M,F)$]\label{op:match}
While there are $x \in M, \ y \in N(M)$ exposed,
we add $\{x,y\}$ to $F$.
\end{operation}

\begin{operation}[{\sc SPLIT}$(M,F)$]\label{op:split}
While there exist $x,x' \in M, \ y,y' \in N(M)$ such that $x$ and $x'$ are exposed, and $\{y,y'\} \in F$,
we replace the edge $\{y,y'\}$ in $F$ by the two new edges $\{x,y\}, \ \{x',y'\}$. 
\end{operation}

We stress that the {\sc MATCH} and {\sc SPLIT} operations correspond to special cases of $F$-augmenting paths, of respective lengthes $1$ and $3$.
Given $F$ and $M$, the ``{\sc MATCH} and {\sc SPLIT}'' technique consists in applying {\sc MATCH}$(M,F)$, which results in a new matching $F'$, then applying {\sc SPLIT}$(M,F')$.
We refer to~\cite{CDP18,FPT97,FGV99,YuY93} for applications of this technique to some graph classes.
In particular, we will use the following lemma for our analysis:

\begin{lemma}[~\cite{CDP18}]\label{lem:join}
Let $G = G_1 \oplus G_2$ be the join of two graphs $G_1,G_2$ and let $F_1,F_2$ be maximum matchings for $G_1,G_2$, respectively.
For $F = F_1 \cup F_2$, applying the ``{\sc MATCH} and {\sc SPLIT}'' technique to $V(G_1)$, then to $V(G_2)$ leads to a maximum matching of $G$.
\end{lemma}

Based on the above lemma, we can readily solve two more cases, namely:

\begin{reduction}\label{rule:universal}[see also~\cite{CDP18}]
Suppose $v_1$ is universal in $G'$.

We set $G^* = G \setminus v_1$, ${\cal P}^* = {\cal P} \setminus \{H_1\}$, ${\cal F}^* = {\cal F} \setminus \{F_1\}$.

\smallskip
Furthermore, let $F^*$ be a maximum matching of the subdivision $G^*({\cal P}^*) = G[V \setminus M_1]$.
We apply Lemma~\ref{lem:join} to $G = G[M_1] \oplus G[V \setminus M_1]$ with $F_1,F^*$ in order to compute a maximum-cardinality matching of $G$.
\end{reduction}

Observe that this above rule is similar to Reduction rule~\ref{rule:isolated}, however it requires an additional post-processing step.

\begin{reduction}\label{rule:true-twin}
Suppose $v_1,v_2$ are true twins in $G'$.
Let $F_2^*$ be the matching of $G[M_1 \cup M_2] = H_1 \oplus H_2$ that is obtained from $F_1,F_2$ by applying Lemma~\ref{lem:join}.

We set $G^* = G \setminus v_1$, ${\cal P}^* = \{H_1 \oplus H_2\} \cup ({\cal P} \setminus \{H_1,H_2\})$, ${\cal F}^* = \{F_2^*\} \cup ({\cal F} \setminus \{F_1,F_2\})$.
\end{reduction}

\paragraph{Complexity analysis.}
Reduction rules~\ref{rule:isolated} and~\ref{rule:false-twin} take constant-time.
The complexity of Reduction rules~\ref{rule:universal} and~\ref{rule:true-twin} is dominated by the {\sc Match} and {\sc Split} operations.
Furthermore, every such operation adds, in constant-time, one or two edges in the matching with exactly one end in $M_1$.
It implies that there cannot be more than ${\cal O}(n_1)$ operations.
Observe that $\Delta m_{G'} - \Delta m_{G^*} \geq n_1n_2 = \Omega(n_1)$.
As a result, Reduction rules~\ref{rule:universal} and~\ref{rule:true-twin} take ${\cal O}(\Delta m_{G'} - \Delta m_{G^*})$-time. 

\subsection{Anti-pendant}\label{sec:anti-pendant}

Suppose $v_1$ is anti-pendant in $G'$.
W.l.o.g., $v_2$ is the unique vertex that is nonadjacent to $v_1$ in $G'$.
By Lemma~\ref{lem:mw-matching-reduction}, we can also assume w.l.o.g. that $E(H_i) = F_i$ for every $i$.
In this situation, we start applying Reduction rule~\ref{rule:isolated}, {\it i.e.}, we set $G^* = G' \setminus v_1$, ${\cal P}^* = {\cal P} \setminus \{H_1\}$, ${\cal F}^* = {\cal F} \setminus \{F_1\}$.
Then, we obtain a maximum-cardinality matching $F^*$ of $G \setminus M_1$ ({\it i.e.}, by applying our reduction rules to this new instance).
Finally, we compute from $F_1,F^*$ a maximum-cardinality matching of $G$, using an intricate procedure.
We detail this procedure next.

\paragraph{First phase: Pre-processing.}
Our correctness proofs in what follows will assume that some additional properties hold on the matched vertices in $F^*$.
So, we start correcting the initial matching $F^*$ so that it is the case.
For that, we introduce two ``swapping'' operations.
Recall that $v_2$ is the unique vertex that is nonadjacent to $v_1$ in $G'$. 

\begin{operation}[{\sc REPAIR}]\label{op:repair}
While there exist $x_2,y_2 \in M_2$ such that $\{x_2,y_2\} \in F_2$ and $y_2$ is exposed in $F^*$,
we replace any edge $\{x_2,w\} \in F^*$ by $\{x_2,y_2\}$.
\end{operation}

This above operation ensures that all the vertices matched by $F_2$ are also matched by $F^*$.
In particular, for every $\{x_2,y_2\} \in F_2$ we have either $\{x_2,y_2\} \in F_2$ or (since we assume $E(H_2) = F_2$) there exist $w,w' \notin M_2$ such that $\{x_2,w\},\{y_2,w'\} \in F^*$.
We stress that this above operation does not modify the cardinality of the matching. 

\begin{operation}[{\sc ATTRACT}]\label{op:attract}
While there exist $x_2\in M_2$ exposed and $\{u,w\} \in F^*$ such that $u \in N_G(M_2), w \notin M_2$, 
we replace $\{u,w\}$ by $\{u,x_2\}$.
\end{operation}

Said otherwise, we give a higher priority on the vertices in $M_2$ to be matched.
Again, we stress that this above operation does not modify the cardinality of the matching. 
Furthermore, Operations~\ref{op:repair} and~\ref{op:attract} are non conflicting since for {\sc Attract} we only consider exposed vertices in $M_2$ ({\it i.e.}, not in $V(F_2)$) and matched edges with their both ends in $N_G(M_1) = V \setminus M_2$.

\smallskip
Let $F^{(0)} = F_1 \cup F^*$.
Summarizing, we get:

\begin{definition}\label{def:good-pty}
A matching $F$ of $G$ is {\em good} if it satisfies the following two properties:
\begin{enumerate}
\item every vertex matched by $F_1 \cup F_2$ is also matched by $F$;
\item either every vertex in $M_2$ is matched, or there is no matched edge in $N_G(M_2) \times N_G(M_1)$.
\end{enumerate}
\end{definition}

\begin{fact}\label{fact:pre-processing}
$F^{(0)}$ is a good matching of $G$.
\end{fact}

\paragraph{Main phase: a modified {\sc Match} and {\sc Split}.}
We now apply the following three operations sequentially:

\begin{enumerate}
\item {\sc Match}$(M_1,F^{(0)})$ (Operation~\ref{op:match}).
Doing so, we obtain a larger matching $F^{(1)}$.

\begin{fact}\label{fact:f-1}
$F^{(1)}$ is a good matching of $G$.
\end{fact}

\begin{proofclaim}
We still have $V(F_1 \cup F_2) \subseteq V(F^{(1)})$ since we only increase the matching by using augmenting paths.
Furthermore, we do not create any new exposed vertex in $M_2$, nor any new matched edge in $N_G(M_2) \times N_G(M_1)$, and so, the second property also stays true.
\end{proofclaim}

\item {\sc Split}$(M_1,F^{(1)})$ (Operation~\ref{op:split}).
Doing so, we obtain a larger matching $F^{(2)}$.

\begin{fact}\label{fact:f-2}
$F^{(2)}$ is a good matching of $G$.
\end{fact}

\item the operation {\sc Unbreak}, defined in what follows (see also Fig.~\ref{fig:unbreak} for an illustration):

\begin{figure}[!h]
\centering
\includegraphics[width=.35\textwidth]{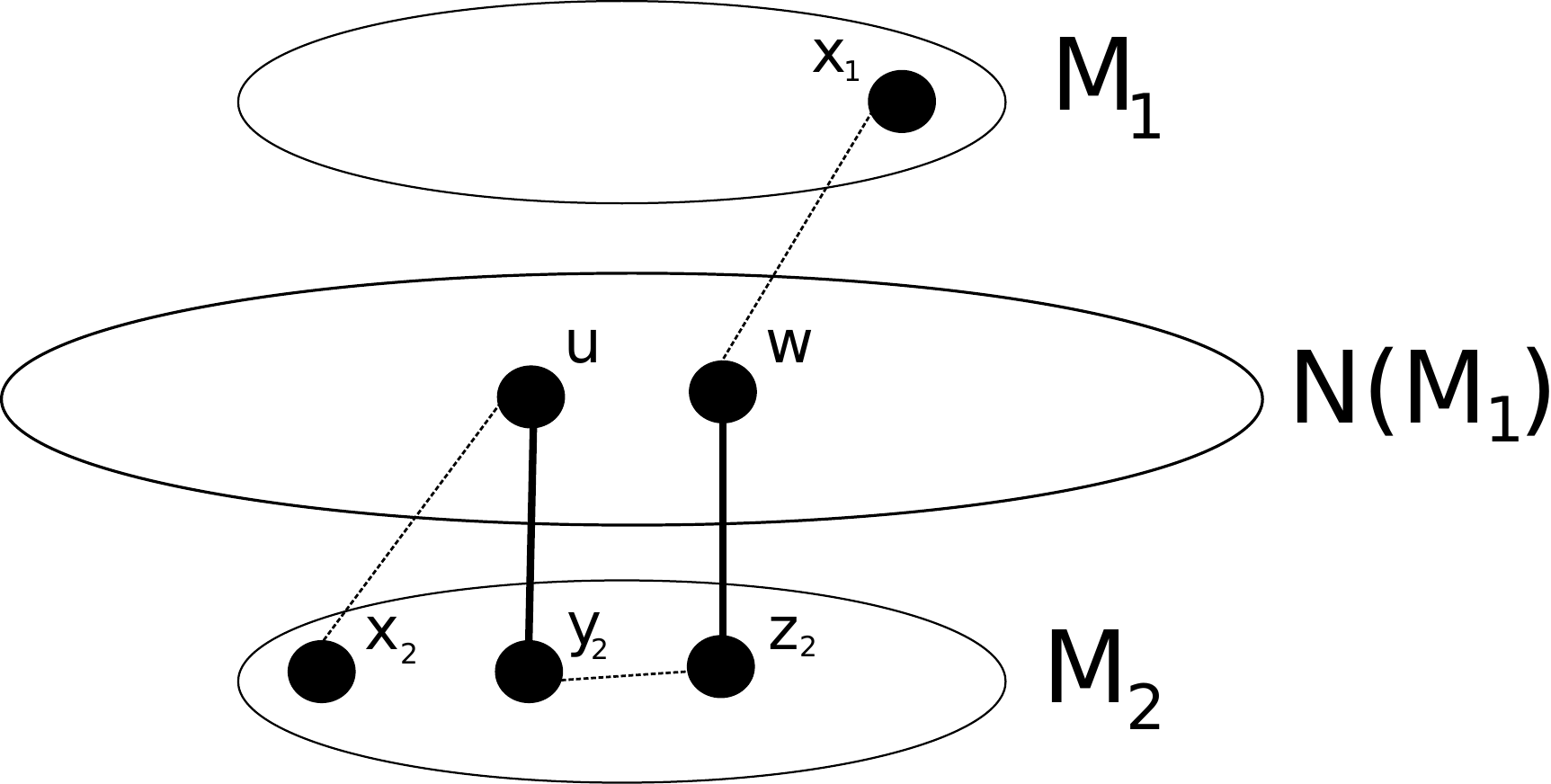}
\caption{An augmenting path of length $5$ with ends $x_1,x_2$. Matched edges are drawn in bold.}
\label{fig:unbreak}
\end{figure}

\begin{operation}[{\sc Unbreak}]\label{op:unbreak}
While there exist $x_1 \in M_1$ and $x_2 \in M_1 \cup M_2$ exposed, and there also exist $\{y_2,z_2\} \in F_2 \setminus F^{(2)}$, we replace any two edges $\{y_2,u\},\{z_2,w\} \in F^{(2)}$ by the three edges $\{x_2,u\},\{y_2,z_2\} \ \mbox{and} \ \{w,x_1\}$.
\end{operation}

We stress that the two edges $\{y_2,u\},\{z_2,w\} \in F^{(2)}$ always exist since $F^{(2)}$ is a good matching of $G$.  
Furthermore doing so, we obtain a larger matching $F^{(3)}$.
\end{enumerate}

The resulting matching $F^{(3)}$ is not necessarily maximum.
However, this matching satisfies the following crucial property:

\begin{lemma}\label{lem:main-anti-pendant}
No vertex of $M_1$ can be an end in an $F^{(3)}$-augmenting path.
\end{lemma}

\begin{proof}
Let $x_1 \in M_1$ be exposed.
Suppose by contradiction $x_1$ is an end of some $F^{(3)}$-augmenting path $P = (x_1 = u_1, u_2, \ldots, u_{2\ell})$.
W.l.o.g., $P$ is of minimum length.
We will derive a contradiction from the following invariants:

\begin{claim}\label{claim:f-3}
The following properties hold for every $0 \leq j \leq 3$:
\begin{enumerate}
\item\label{pty-1} $F^{(j)}$ is a good matching of $G$;
\item\label{pty-2} If $u_{2i},u_{2i+1} \in N_G(M_1)$ and $\{u_{2i},u_{2i+1}\} \in F^{(3)}$ then we also have $\{u_{2i},u_{2i+1}\} \in F^{(j)}$; 
\item\label{pty-3} $F_1 \subseteq F^{(j)}$.
\end{enumerate}
\end{claim}

\begin{proofclaim}
The proof readily follows from the same arguments as for Fact~\ref{fact:f-1}.
Since we only increase the successive matchings using augmenting paths, we keep the property that $V(F_1 \cup F_2) \subseteq V(F^{(j)})$.
In fact, since we only consider the exposed vertices in $M_1$ for our operations, we have the stronger Property~\ref{pty-3} that $F_1 \subseteq F^{(j)}$.
Furthermore, our successive operations do not create any new exposed vertex in $M_2$ nor any new matched edge in $N_G(M_1) \times N_G(M_1)$, and so, both Properties~\ref{pty-1} and~\ref{pty-2} also hold.
\end{proofclaim}

In what follows, we divide the proof in three claims.

\begin{claim}\label{claim:in-m1-m2}
$u_{2\ell} \in M_1 \cup M_2$.
\end{claim}

\begin{proofclaim}
Suppose for the sake of contradiction $u_{2\ell} \notin M_1 \cup M_2$, or equivalently $u_{2\ell} \in N(M_1)$.
Then, we could have continued the first step {\sc Match}$(M_1,F^{(0)})$ by matching $x_1,u_{2\ell}$ together, that is a contradiction.
\end{proofclaim}

Next, we derive a contradiction by proving $u_{2\ell} \notin M_1 \cup M_2$.

\begin{claim}\label{claim:notin-m1}
$u_{2\ell} \notin M_1$.
\end{claim}

\begin{proofclaim}
Suppose for the sake of contradiction $u_{2\ell} \in M_1$.
There are two cases.
\begin{enumerate}

\item
Case $u_2,u_3 \in N(M_1)$.
Since $\{u_2,u_3\} \in F^{(3)}$ we have by Claim~\ref{claim:f-3} $\{u_2,u_3\} \in F^{(1)}$.
In particular, we could have replaced $\{u_2,u_3\}$ by $\{u_2,x_1\},\{u_3,u_{2\ell}\}$ during the second step of the main phase ({\it i.e.}, {\sc Split}$(M_1,F^{(1)})$), that is a contradiction.

\item
Thus, let us now assume $u_2 \in N(M_1)$ (necessarily) but $u_3 \notin N(M_1)$.
By minimality of $P$, we have $u_4 \notin N(M_1)$ (otherwise, $P'=(x_1,u_4,u_5,\ldots,u_{2\ell})$ would be a shorter augmenting path than $P$).
We claim that it implies $u_3 \notin M_1$.
Indeed, otherwise we should also have $u_4 \in M_1$, and so, $\{u_3,u_4\} \in F_1$ since we assume that $M_1$ induces a matching.
However, $\{u_3,u_4\} \notin F^{(3)}$, whereas we have by Claim~\ref{claim:f-3} that $F_1 \subseteq F^{(3)}$.
A contradiction.
Therefore, as claimed, $u_3 \notin M_1$.
We deduce from the above that $u_3 \in M_2$.
Similarly, since $u_4 \notin N(M_1) \supseteq N(M_2)$ we get $u_4 \in M_2$.
Altogether combined (and since $M_2$ induces a matching), $\{u_3,u_4\} \in F_2 \setminus F^{(3)}$.
Then, we could have continued the step {\sc Unbreak} with $x_1 \in M_1$ and $u_{2\ell} \in M_1$ exposed, and $\{u_3,u_4\} \in F_2 \setminus F^{(3)}$, that is a contradiction.  
\end{enumerate}
As a result, $u_{2\ell} \notin M_1$.
\end{proofclaim}

\begin{claim}\label{claim:notin-m2}
$u_{2\ell} \notin M_2$.
\end{claim}

\begin{proofclaim}
Suppose for the sake of contradiction $u_{2\ell} \in M_2$.
First we prove $u_2,u_3 \in N(M_1)$. 
Indeed, since $u_1$ is exposed and $F_1 \subseteq F^{(3)}$ by Claim~\ref{claim:f-3} we have that $u_2 \in N(M_1)$.
Furthermore, if $u_3 \notin N(M_1)$ then we could prove as before (Claim~\ref{claim:notin-m1}, Case 2) $\{u_3,u_4\} \in F_2 \setminus F^{(3)}$; it implies that we could have continued the step {\sc Unbreak} with $x_1 \in M_1$ and $u_{2\ell} \in M_2$ exposed, and $\{u_3,u_4\} \in F_2 \setminus F^{(3)}$, that is a contradiction. 
Therefore, as claimed, $u_2,u_3 \in N(M_1)$.  

Now, consider the edge $\{u_{2\ell-2},u_{2\ell-1}\} \in F^{(3)}$.
Since $u_{2\ell} \in M_2$ is exposed and by Claim~\ref{claim:f-3} we have $V(F_2) \subseteq V(F^{(3)})$, $u_{2\ell} \notin V(F_2)$.
Furthermore, since $E(H_2) = F_2$ we have $u_{2\ell-1} \notin M_2$.
There are two cases.
\begin{enumerate}

\item 
Suppose $u_{2\ell-2} \in M_1$.
Since $u_{2\ell-2}$ is matched to $u_{2\ell-1} \notin M_1$ and $F_1 \subseteq F^{(3)}$ by Claim~\ref{claim:f-3}, we have that $u_{2\ell-2} \notin V(F_1)$.
Furthermore, we claim that the edge $\{u_{2\ell-2},u_{2\ell-1}\}$ was added to the matching during the second step of the main phase ({\it i.e.}, {\sc Split}$(M_1,F^{(1)})$).

In order to prove this subclaim, we only need to decide the first step where $u_{2\ell-2}$ was matched to any vertex; indeed, our above operations can only consider vertices in $M_1$ that are exposed.
We observe that $u_{2\ell-2},u_{2\ell-1}$ could not possibly be matched together during the first step since otherwise, we could have also matched $u_{2\ell-1}$ with $u_{2\ell}$, thereby contradicting that $F^*$ is a maximum-cardinality matching of $G \setminus M_1$.
In addition, recall that we proved above $u_2,u_3 \in N(M_1)$.
It implies that $u_{2\ell-2},u_{2\ell-1}$ were matched together during the second step since, otherwise, this second step could have continued with $x_1,u_{2\ell-2} \in M_1$ exposed and $\{u_2,u_3\} \in F^{(1)}$.
Therefore, this claim is proved.

Then, before the second step of the main phase happened, vertex $u_{2\ell-1}$ was matched to some other vertex in $N_G(M_1)$.
However, since $u_{2\ell-1} \in N_G(M_2)$ and $u_{2\ell} \in M_2$ is exposed the latter contradicts that $F^{(1)}$ is good, and so, Claim~\ref{claim:f-3}.

\item Thus from now on assume $u_{2\ell-2} \notin M_1$.
By Claim~\ref{claim:f-3} we have that $F^{(3)}$ is good, and so, since $u_{2\ell-1} \in N_G(M_2)$ and $u_{2\ell} \in M_2$ is exposed we have $u_{2\ell-2} \in M_2$.
Furthermore, $u_{2\ell-3} \notin M_2$ since, otherwise, the final step of the main phase ({\sc Unbreak}) could have continued with $x_1 \in M_1$ and $u_{2\ell} \in M_2$ exposed, and $\{u_{2\ell-3},u_{2\ell-2}\} \in F_2 \setminus F^{(3)}$, that is a contradiction.  
However, it implies that $P'=(x_1 = u_1, u_2, u_3, \ldots, u_{2\ell-3},u_{2\ell})$ is a shorter augmenting path than $P$, thereby leading to another contradiction.
\end{enumerate}
As a result, $u_{2\ell} \notin M_2$.
\end{proofclaim}

Overall since $u_{2\ell} \in M_1 \cup M_2$ by Claim~\ref{claim:in-m1-m2} but $u_{2\ell} \notin M_1 \cup M_2$ by Claims~\ref{claim:notin-m1} and~\ref{claim:notin-m2}, the above proves that $x_1$ cannot be an end in any $F^{(3)}$-augmenting path.
\end{proof}

\paragraph{Finalization phase: breaking some edges in $F_1$.}
Intuitively, the matching $F^{(3)}$ may not be maximum because we sometimes need to borrow some edges of $F_1$ in augmenting paths.
So, we complete our procedure by performing the following two operations:

\begin{itemize}
\item Let $U_1 = N(M_1) \setminus V(F^{(3)})$, {\it i.e.}, $U_1$ contains all the exposed vertices in $N(M_1)$.
Consider the subgraph $G[M_1 \cup U_1] = G[M_1] \oplus G[U_1]$.
The set $U_1$ is a module of this subgraph.
We apply {\sc Split}$(U_1,F^{(3)})$ in $G[M_1 \cup U_1]$.
Doing so, we obtain a larger matching $F^{(4)}$.

\begin{fact}\label{fact:f-4}
$F^{(4)}$ is a good matching of $G$.
\end{fact}

\item Then, we apply the operation {\sc LocalAug}, defined next (see also Fig.~\ref{fig:local-aug} for an illustration):

\begin{figure}[h!]
\centering
\includegraphics[width=.35\textwidth]{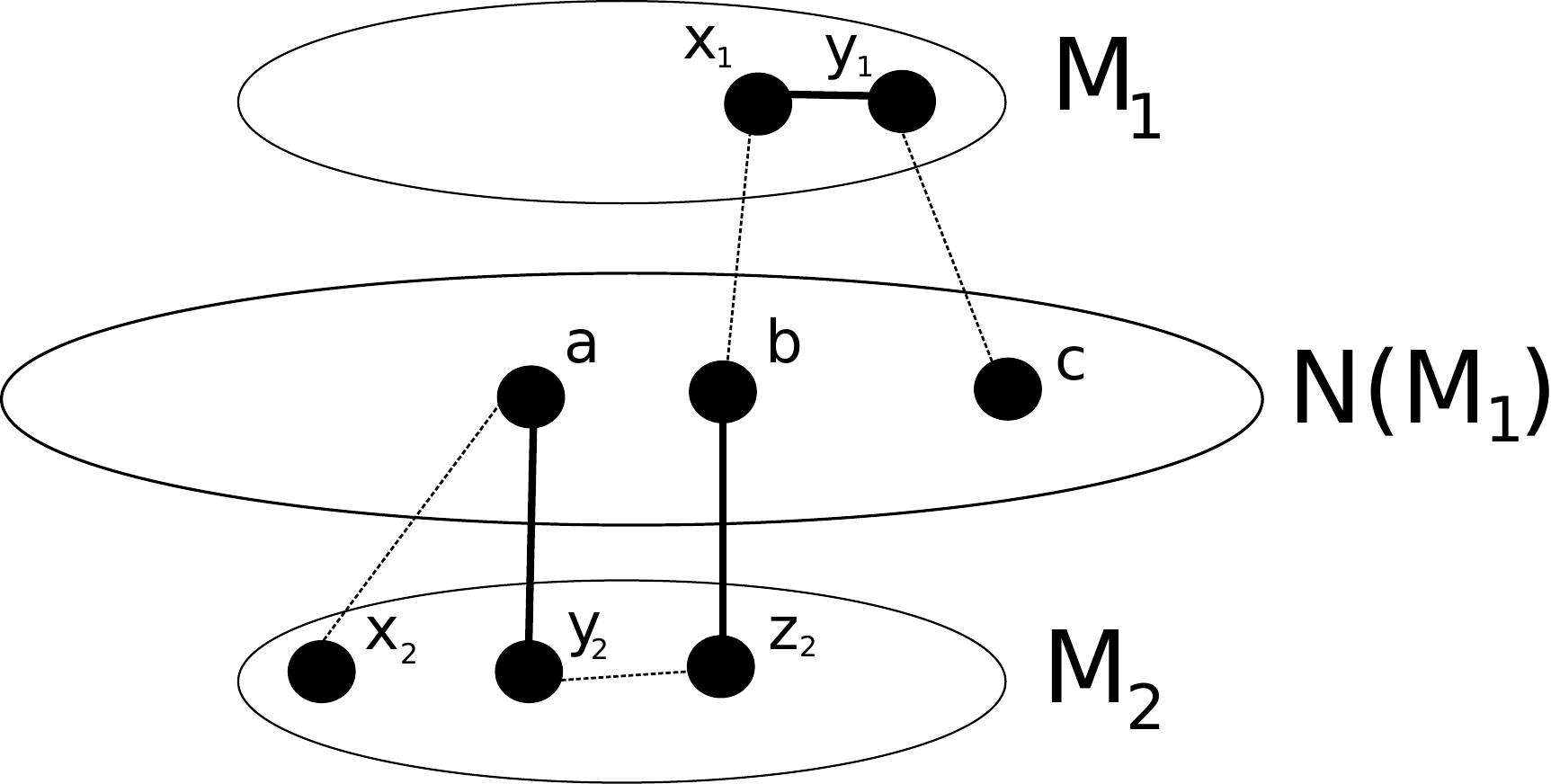}
\caption{An augmenting path of length $7$ with ends $x_2,c$. Matched edges are drawn in bold.}
\label{fig:local-aug}
\end{figure}

\begin{operation}[{\sc LocalAug}]\label{op:local-aug}
While there exist $x_2 \in M_2$ and $c \in N(M_1)$ exposed, and there also exist $\{x_1,y_1\} \in F^{(4)}$ and $\{y_2,z_2\} \in F_2 \setminus F^{(4)}$, we do the following:
\begin{itemize}
\item we remove $\{x_1,y_1\}$ and any edge $\{a,y_2\}, \{b,z_2\}$ from $F^{(4)}$;
\item we add $\{x_2,a\},\{y_2,z_2\},\{b,x_1\} \ \mbox{and} \ \{y_1,c\}$ in $F^{(4)}$.
\end{itemize}
\end{operation}

We stress that the two edges $\{y_2,a\},\{z_2,b\} \in F^{(4)}$ always exist since $F^{(4)}$ is a good matching of $G$.  
Furthermore doing so, we obtain a larger matching $F^{(5)}$.
\end{itemize}

\begin{lemma}\label{lem:final-anti-pendant}
$F^{(5)}$ is a maximum-cardinality matching of $G$.
\end{lemma}

\begin{proof}
Suppose for the sake of contradiction that there exists an $F^{(5)}$-augmenting path $P=(u_1,u_2,\ldots,u_{2\ell})$.
W.l.o.g., $P$ is of minimum size.
We divide the proof into the following claims.

\begin{fact}\label{fact:f-5}
$F^{(5)}$ is a good matching of $G$.
\end{fact}

\begin{claim}\label{claim:not-m1}
$u_1,u_{2\ell} \notin M_1$
\end{claim}

\begin{proofclaim}
Suppose for the sake of contradiction $u_1 \in M_1$ (the case $u_{2\ell} \in M_1$ is symmetrical to this one).
Since $F^{(3)}$ and $F^{(5)} \oplus P$ are matchings, the symmetric difference $F^{(3)} \oplus (F^{(5)} \oplus P)$ is a disjoint union of alternating cycles, alternating paths and isolated vertices.
In particular, since $F^{(5)} \oplus P$ can be obtained from $F^{(3)}$ by using augmenting paths, the symmetric difference $F^{(3)} \oplus (F^{(5)} \oplus P)$ is exactly a disjoint union of isolated vertices and of augmenting paths that can be used for obtaining $F^{(5)} \oplus P$ from $F^{(3)}$. 
One of these paths must contain $u_1$.
As a result, there is also an $F^{(3)}$-augmenting path with an end in $M_1$, thereby contradicting Lemma~\ref{lem:main-anti-pendant}.
The latter proves, as claimed, $u_1,u_{2\ell} \notin M_1$.
\end{proofclaim}

\begin{claim}\label{claim:not-nm1}
There is no exposed vertex in $N(M_1)$.
\end{claim}

\begin{proofclaim}
Suppose for the sake of contradiction $N(M_1) \not\subseteq V(F^{(5)})$.
In particular $N(M_1) \not\subseteq V(F^{(4)})$.
We will prove in this situation there can be only one vertex in $N(M_1)$ that is left exposed by $F^{(4)}$.
Then, we will derive a contradiction by proving that we can apply the {\sc LocalAug} operation.

First, we observe that the main phase of our procedure must terminate after its very first step {\sc Match}$(M_1,F^{(0)})$.
Indeed, after this step there can be no more exposed vertex in $M_1$, and so, the other rules cannot apply.
Then, we apply the operation {\sc Split}$(U_1,F^{(3)})$ in $G[M_1 \cup U_1]$ in order to further match some vertices in $N(M_1)$ to the vertices in $M_1$.
Doing so, we get the following two important properties for $F^{(4)}$:
\begin{enumerate}
\item\label{pty-1} if a vertex of $N(M_1)$ is matched to a vertex of $M_1$, then this vertex was left exposed by $F^*$;
\item\label{pty-2} $F^* \subseteq F^{(4)}$.
\end{enumerate}
Let $Q = (w_1,w_2,\ldots,w_{2q})$ be a minimum-length $F^{(4)}$-augmenting path.
Such path exist since $F^{(5)}$, and so, $F^{(4)}$, is not maximum.
Let $i_0$ be the minimum index $i$ such that $w_i \in M_1$.
The latter is well-defined since otherwise, by the above Property~\ref{pty-2} $Q$ would be an $F^*$-augmenting path in $G \setminus M_1$, thereby contradicting the maximality of $F^*$.
Furthermore, $w_1,w_{2q} \notin M_1$, and so, $i_0 > 1$ (the proof is the same as for Claim~\ref{claim:not-m1}).
More generally, we have that $i_0$ is even since, if it were not the case, by the above Property~\ref{pty-1} we would have that $(w_1,w_2,\ldots,w_{i_0-1})$ is an $F^*$-augmenting path in $G \setminus M_1$, thereby again contradicting the maximality of $F^*$.
There are two cases:

\begin{enumerate}

\item
Case there exists an edge $\{x_1,y_1\} \in F_1 \cap F^{(4)}$.
Then, there is exactly one exposed vertex $c \in N_G(M_1)$ (otherwise, the step {\sc Split}$(U_1,F^{(3)})$ in $G[M_1 \cup U_1]$ could have been continued).
W.l.o.g., $w_1 \neq c$.
Since $w_1 \notin M_1$, it implies $w_1 \in M_2$.
We consider the alternating subpath $(w_1,w_2,w_3)$ in $Q$.
Since $w_2 \notin M_1$ and $i_0$ is even, we have $w_3 \notin M_1$ by minimality of $i_0$.
Furthermore, since $w_1 \in M_2$ is exposed, $w_2 \in N(M_2)$ is matched to $w_3$ and $F^{(4)}$ is good by Fact~\ref{fact:f-4}, $w_3 \notin N(M_1)$.
Altogether combined, $w_3 \in M_2$.
We also have $w_4 \in M_2$ since otherwise, $Q'=(w_1,w_4,w_5,\ldots,w_{2q})$ would be a shorter augmenting path than $Q$, thereby contradicting the minimality of $Q$.
As a result, $\{w_3,w_4\} \in F_2 \setminus F^{(4)}$.
However, in this case the step {\sc LocalAug} will be applied, {\it e.g.} with $w_1 \in M_2$ and $c \in N(M_1)$ exposed, $\{x_1,y_1\} \in F_1 \cap F^{(4)}$ and $\{w_3,w_4\} \in F_2 \setminus F^{(4)}$.
In particular, $c$ is matched by $F^{(5)}$, that is a contradiction.

\item
Case $F_1 \cap F^{(4)} = \emptyset$.
In particular, $w_{i_0+1} \notin M_1$.
Let $j_0$ be the maximum $j \geq i_0+1$ such that $w_{i_0+1}, w_{i_0+2}, \ldots, w_j \notin M_1$.
We have $j_0 < 2q$ since otherwise, by the above Property~\ref{pty-1} $(w_{i_0+1},\ldots,w_{2q})$ would be an $F^*$-augmenting path in $G \setminus M_1$, thereby contradicting the maximality of $F^*$.
Thus, $w_{j_0+1} \in M_1$.
Furthermore, $j_0$ is even since otherwise, $Q'=(w_1,\ldots,w_{i_0-1},w_{j_0+1},\ldots,w_{2q})$ would be a shorter $F^{(4)}$-augmenting path than $Q$, thereby contradicting the minimality of $Q$.
However, then we have by the above Property~\ref{pty-1} $(w_{i_0+1},\ldots,w_{j_0})$ that is an $F^*$-augmenting path in $G \setminus M_1$, thereby contradicting the maximality of $F^*$.
\end{enumerate}
Overall, the above proves as claimed that there is no exposed vertex in $N(M_1)$.
\end{proofclaim}

It follows from Claims~\ref{claim:not-m1} and~\ref{claim:not-nm1} that $u_1,u_{2\ell} \in M_2$.
Furthermore we have $u_1 \in M_2$ is exposed, $\{u_2,u_3\}$ is matched and $u_2 \in N_G(M_2)$.
Since $F^{(5)}$ is good by Fact~\ref{fact:f-5} we have $u_3 \notin N(M_1)$.
Equivalently, $u_3 \in M_1 \cup M_2$.
The following claim will be instrumental in deriving a contradiction.

\begin{claim}\label{claim:broken-edge}
$F_2 \subseteq F^{(5)}$.
\end{claim}

\begin{proofclaim}
Suppose for the sake of contradiction there exists $\{x_2,y_2\} \in F_2 \setminus F^{(5)}$.
We prove that $F_2 \setminus F^{(5)} \subseteq F_2 \setminus F^*$.
Indeed, after the two first steps of the main phase we have $F_2 \setminus F^{(2)} = F_2 \setminus F^*$.
The operation {\sc Unbreak} adds edges of $F_2$ into the matching, hence $F_2 \setminus F^{(3)} \subseteq F_2 \setminus F^*$.
Then, after the operation {\sc Split}$(U_1,F^{(3)})$ in $G[M_1 \cup U_1]$ we have $F_2 \setminus F^{(4)} = F_2 \setminus F^{(3)} \subseteq F_2 \setminus F^*$.
Finally, the operation {\sc LocalAug} adds edges of $F_2$ into the matching, hence $F_2 \setminus F^{(5)} \subseteq F_2 \setminus F^*$.

However, since $F^{(0)}$ is good, we have $V(F_2) \subseteq V(F^*)$.
It implies there exist $w,w' \in N(M_2)$ such that $\{x_2,w\},\{y_2,w'\} \in F^*$.
In particular we have that $(u_1,w,x_2,y_2,w',u_{2\ell})$ is an $F^*$-augmenting path in $G \setminus M_1$, thereby contradicting the maximality of $F^*$.
\end{proofclaim}

Now, there are two cases.
\begin{itemize}
\item Case $u_3 \in M_2$.
We have $u_4 \notin N(M_2)$ since otherwise, $P'=(u_1,u_4,u_5,\ldots,u_{2\ell})$ would be a shorter augmenting path than $P$, thereby contradicting the minimality of $P$.
Therefore, $\{u_3,u_4\} \in F_2 \setminus F^{(5)}$.
The latter contradicts Claim~\ref{claim:broken-edge}.

\item Case $u_3 \in M_1$.
By maximality of $F^*$, $u_2$ was matched in $F^*$ (otherwise, we could have added $\{u_1,u_2\}$ in $F^*$).
Therefore, the edge $\{u_2,u_3\}$ was not matched during the operation {\sc Match}$(M_1,F^{(0)})$ nor during the operation {\sc Split}$(U_1,F^{(3)})$ in $G[M_1 \cup U_1]$.
Furthermore, this edge was not matched during the operation {\sc Split}$(M_1,F^{(1)})$ either since otherwise, $u_2$ would have been matched in $F^{(1)}$ with some other vertex in $N(M_1)$; since $u_1 \in M_2$ is exposed and $u_2 \in N(M_2)$, the latter would contradict that $F^{(1)}$ is good (Fact~\ref{fact:f-1}).
As a result, the edge $\{u_2,u_3\}$ was matched during the {\sc Unbreak} operation or the {\sc LocalAug} operation.
Both subcases imply the existence of some edge $\{x_2,y_2\} \in F_2 \setminus F^*$.
As in the previous case, the latter contradicts the maximality of $F^*$.
\end{itemize}
\end{proof}

\paragraph{Complexity analysis.}
Each step of our procedure is corresponding to a while loop.
In order to execute any loop in constant-time we need constant-time access to the following objects: 
\begin{itemize}
\item exposed vertices in $M_1, M_2$ or $N(M_1)$.
Recall that we have access to a canonical ordering for every module ({\it i.e.}, see Section~\ref{sec:prelim}).
Hence, constant-time access to the exposed vertices can be ensured up to ${\cal O}(p)$-time preprocessing, where $p = |V(G')|$.
\item matched edges with at least one end in $N(M_1)$.
We assume that for every matched vertex $u$, we can output in constant-time the unique edge of the matching that contains $u$.
Thus, constant-time access to these matched edges can be ensured up to ${\cal O}(|N_G(M_1)|)$-time.

We also need constant-time access to the subset of these matched edges that have their other end: also in $N_G(M_1)$; in $N_G(M_2)$; or in $V(F_2)$.
This takes additional ${\cal O}(|N_G(M_1)|)$-time preprocessing.
\item matched edges in $F_1$.
For that, it suffices to scan the canonical ordering of $M_1$, that takes ${\cal O}(n_1)$-time.
\item finally, unmatched edges in $F_2$.
For that, we enumerate the matched edges with their ends in $N_G(M_2)$ and $M_2$.
Doing so, since $E(H_2) = F_2$, we can enumerate for every such end in $M_2$ the unique unmatched edge in $F_2$ to which it is incident.
\end{itemize} 
Altogether combined, after a pre-processing in time ${\cal O}(|N_G[M_1]|)$, any loop of the procedure can be executed in constant-time.
Note that ${\cal O}(|N_G[M_1]|) = {\cal O}(\delta_1) = {\cal O}(\Delta m(G') - \Delta m(G^*))$.

We observe that after any loop, a new edge is added to the matching with exactly one end in $M_1$.
Furthermore, this edge is never removed from the matching at an ulterior step.
Hence, the total number of loops is an ${\cal O}(|N_G[M_1]|)$.
Overall, the total running-time of the procedure is in ${\cal O}(|N_G[M_1]|)$, that is in ${\cal O}(\Delta m(G') - \Delta m(G^*))$.

\subsection{Pendant}\label{sec:pendant}

Suppose $v_1$ is pendant in $G'$.
W.l.o.g., $v_2$ is the unique vertex that is adjacent to $v_1$ in $G'$.
This last case is arguably more complex than the others since it requires both a pre-processing and a post-processing treatment on the matching.

We note that a particular subcase was solved in~\cite{CDP18}, namely, when $M_2$ is a trivial module.
However, our techniques for the general case are quite different than the techniques in~\cite{CDP18}. 

\paragraph{First phase: greedy matching.}
We apply the {\sc Match} \& {\sc Split} technique to $M_1$.
Doing so, we obtain a set $F_{1,2}$ of matched edges between $M_1$ and $M_2$.
We remove $V(F_{1,2})$, the set of vertices incident to an end of $F_{1,2}$, from $G$.
Then, two situations can occur.
In the first situation, this initial pre-treatment suffices in order to prune $M_1$ (pathological cases).
Otherwise, at most one exposed vertex remains in $M_1$; we arbitrarily break an edge of $F_2$ to match such vertex.
More precisely, there are three cases.
\begin{itemize}
\item
If $M_2 \subseteq V(F_{1,2})$ then $M_1 \setminus V(F_{1,2})$ is now an isolated module.
We can apply Reduction rule~\ref{rule:isolated}.
\item
If $M_1 \subseteq V(F_{1,2})$ then $M_1$ is already eliminated.
Let $F_2^*$ contain the edges of $F_2$ that are not incident to a vertex of $M_2 \cap V(F_{1,2})$. 

We set $G^* = G' \setminus v_1$, ${\cal P}^* = \{H_2 \setminus V(F_{1,2})\} \cup ({\cal P} \setminus \{H_1,H_2\}), {\cal F}^* = \{F_2^*\} \cup ({\cal F} \setminus \{F_1,F_2\})$.  
\item The interesting case is when both $M_1 \setminus V(F_{1,2})$ and $M_2 \setminus V(F_{1,2})$ are nonempty.
In particular, suppose there remains an exposed vertex $x_1 \in M_1 \setminus V(F_{1,2})$.
Since $M_2 \setminus V(F_{1,2}) \neq \emptyset$, there exists $\{x_2,y_2\} \in F_2$ such that $x_2,y_2 \notin V(F_{1,2})$.
We remove $x_1$ from $M_1$, $x_2$ from $M_2$, $\{x_2,y_2\}$ from $F_2$ and then we add $\{x_1,x_2\}$ in $F_{1,2}$. 
\end{itemize}

Our first result in this section is that there always exists an optimal solution that contains $F_{1,2}$.
This justifies a posteriori the removal of $V(F_{1,2})$ from $G$.

\begin{lemma}\label{lem:greedy-pendant}
There is a maximum-cardinality matching of $G$ that contains all edges in $F_{1,2}$.
\end{lemma}

\begin{proof}
Let $M_1 = (u_1,u_2,\ldots,u_{n_1})$ and $M_2 = (w_1,w_2,\ldots,w_{n_2})$ be canonically ordered w.r.t. $F_1,F_2$ (cf. Sec.~\ref{sec:prelim}).
Furthermore, let $u_1,u_2,\ldots,u_k$ be the maximal sequence of exposed vertices in $M_1$ with $k \leq n_2$.
We observe that $F_{1,2}$ is obtained by greedily matching $u_i$ with $w_i$.

Then, let $F$ be any maximum-cardinality matching of $G$ that can be obtained from $F_{1,2}$ using augmenting paths.
By construction, $u_1,u_2,\ldots,u_k$ are matched by $F$.
In particular, since every $u_i$ is isolated in $H_1 = G[M_1]$, it is matched by $F$ to some vertex in $M_2$.
So, let $A_2 \subseteq M_2$ be the vertices matched by $F$ with a vertex in $V \setminus M_2$ (possibly, in $M_1$).
Since $M_2$ is a module, we can always assume that $A_2$ induces a suffix $(w_1,w_2,\ldots,w_j)$ of the canonical ordering ({\it i.e.}, see~\cite[Lemma $5.1$]{CDP18}).
Finally, let $B_2 \subseteq V \setminus M_2$, $|B_2| = |A_2|$, be the set of vertices matched by $F$ with a vertex of $A_2$.
Note that we have $u_1,u_2,\ldots,u_k \in B_2$.
Since $M_2$ is a module, there are all possible edges between $A_2$ and $B_2$.
As a result, we can always replace the matched edges between $A_2,B_2$ by any perfect matching between these two sets without changing the cardinality of $F$.
It implies that we can assume w.l.o.g. every $u_i$ is matched to $w_i$. 
\end{proof}

We stress that during this phase, all the operations except maybe the last one increase the cardinality of the matching.
Furthermore, the only possible operation that does not increase the cardinality of the matching is the replacement of an edge in $F_2$ by an edge in $F_{1,2}$.
Doing so, either we fall in one of the two pathological cases $M_1 \subseteq V(F_{1,2})$ or $M_2 \subseteq V(F_{1,2})$ (easy to solve), or then we obtain through the replacement operation the following stronger property: 

\begin{property}\label{pty:pendant}
All vertices in $M_1$ are matched by $F_1$.
\end{property} 

We will assume Property~\ref{pty:pendant} to be true for the remaining of this section.

\paragraph{Second phase: virtual split edges.}
We complete the previous phase by performing a {\sc Split} between $M_2,M_1$ (Operation~\ref{op:split}).
That is, while there exist two exposed vertices $x_2,y_2 \in M_2$ and a matched edge $\{x_1,y_1\} \in F_1$ we replace $\{x_1,y_1\}$ by $\{x_1,x_2\},\{y_1,y_2\}$ in the current matching. 
However, we encode the {\sc Split} operation using virtual edges in $H_2$.

Formally, we add a virtual edge $\{x_2,y_2\}$ in $H_2$ that is labeled by the corresponding edge $\{x_1,y_1\} \in F_1$.
Let $H_2^*$ and $F_2^*$ be obtained from $H_2$ and $F_2$ by adding all the virtual edges. 
We set $G^* = G' \setminus v_1$, ${\cal P}^* = \{H_2^*\} \cup ({\cal P} \setminus \{H_1,H_2\})$ and ${\cal F}^* = \{F_2^*\} \cup ({\cal F} \setminus \{F_1,F_2\})$.

Intuitively, the virtual edges are used in order to shorten the augmenting paths crossing $M_1$.

\paragraph{Third phase: post-processing.}
Let $F^*$ be a maximum-cardinality matching of the subdivision $G^*({\cal P}^*)$ ({\it i.e.}, obtained by applying our reduction rules to the new instance).
We construct a matching $F$ for $G$ as follows.
\begin{enumerate}
\item
We add in $F$ all the non virtual edges in $F^*$.
\item
For every virtual edge $\{x_2,y_2\}$, let $\{x_1,y_1\} \in F_1$ be its label.
If $\{x_2,y_2\} \in F^*$ then we add $\{x_1,y_2\},\{x_2,y_1\}$ in $F$, otherwise we add $\{x_1,y_1\}$ in $F$.
In the first case, we say that we confirm the {\sc Split} operation, whereas in the second case we say that we cancel it.
\item
Finally, we complete $F$ with all the edges of $F_1$ that do not label any virtual edge ({\it i.e.}, unused during the second phase).
\end{enumerate}

\begin{lemma}\label{lem:main-pendant}
$F$ is a maximum-cardinality matching of $G$.
\end{lemma}

\begin{proof}
Suppose for the sake of contradiction that $F$ is not maximum.
Let $P=(u_1,u_2,\ldots,u_{2\ell})$ be a shortest $F$-augmenting path.
In order to derive a contradiction, we will transform $P$ into an $F^*$-augmenting path in $G^*({\cal P}^*)$.
For that, we essentially need to avoid passing by $M_1$, using instead the virtual edges.
In the first part of the proof, we show that $P$ intersects $M_1$ in at most one edge (Claim~\ref{claim:small-cross}).
We need a few preparatory claims in order to prove this result.

First we observe that the two ends of $P$ cannot be in $M_1$:

\begin{claim}\label{claim:m1-matched}
$M_1 \subseteq V(F)$.
In particular, $u_1,u_{2\ell} \notin M_1$.
\end{claim}

\begin{proofclaim}
According to Property~\ref{pty:pendant}, all vertices in $M_1$ are matched by $F_1$.
Our procedure during the third phase ensures that $V(F_1) \subseteq V(F)$, and so, $M_1 \subseteq V(F)$. 
\end{proofclaim}

Then, we prove that for every $\{x_1,y_1\} \in F$ we have either $x_1,y_1 \notin V(P)$ or $\{x_1,y_1\} \in E(P)$.
This result follows from the combination of Claims~\ref{claim:no-half-f1} and~\ref{claim:respect-f1}.

\begin{claim}\label{claim:no-half-f1}
Let $\{x_1,y_1\} \in F$.
Either $x_1,y_1 \in V(P)$ or $x_1,y_1 \notin V(P)$.
\end{claim}

\begin{proofclaim}
Suppose for the sake of contradiction $x_1 \in V(P)$ but $y_1 \notin V(P)$.
Up to reverting the path $P$ we have $x_1 = u_{2i+1}$ for some $i$.
Then, since we have $y_1 \notin V(P)$ and $M_1$ induces a matching, $u_{2i+2} \notin M_1$.
It implies $u_{2i+2} \in M_2$.
Furthermore, our construction ensures that $u_{2i}$ (the vertex matched with $x_1$) was left exposed by $F_2$.
Indeed, $u_{2i}$ must be an end of a virtual edge (cf. Second phase).
Since $E(H_2) = F_2$ it implies $u_{2i-1} \notin M_2$. 
Finally, since $u_{2i-1} \in N_G(M_2)$ and $M_2$ is a module, $P'=(u_1,u_2,\ldots,u_{2i-1},u_{2i+2},\ldots,u_{2\ell})$ is a shorter augmenting path than $P$, thereby contradicting the minimality of $P$.
\end{proofclaim}

\begin{claim}\label{claim:respect-f1}
Let $u_i,u_j \in V(P) \cap M_1, \ j > i$, such that $\{u_i,u_j\} \in F_1$.
Then, $j = i+1$.
\end{claim}

\begin{proofclaim}
The result trivially holds if $\{u_i,u_j\} \in F$.
Thus, we assume from now on $\{u_i,u_j\} \notin F$.
We need to consider the following cases:
\begin{itemize}
\item Case $i$ odd, $j$ even.
Since $P'=(u_1,u_2,\ldots,u_{i-1},u_i,u_j,u_{j+1},\ldots,u_{2\ell})$ is also an augmenting path, we get $j=i+1$ by minimality of $P$.
\item Case $i$ odd, $j$ odd.
Note that $u_{j+1} \notin M_1$ since we assume $\{u_i,u_j\} \in F_1$ and $M_1$ induces a matching.
Then, since $u_{j+1} \in N_G(M_1)$ and $M_1$ is a module we have that $P'=(u_1,u_2,\ldots,u_{i-1},u_i,u_{j+1},\ldots,u_{2\ell})$ is a shorter augmenting path than $P$, thereby contradicting the minimality of $P$.
\item Case $i$ even, $j$ even.
Note that $u_{i-1} \notin M_1$ since we assume $\{u_i,u_j\} \in F_1$ and $M_1$ induces a matching.
Then, since $u_{i-1} \in N_G(M_1)$ and $M_1$ is a module we have that $P'=(u_1,u_2,\ldots,u_{i-1},u_j,u_{j+1},\ldots,u_{2\ell})$ is a shorter augmenting path than $P$, thereby contradicting the minimality of $P$.
\item Case $i$ even, $j$ odd.
As before, we have $u_{i-1},u_{j+1} \notin M_1$, that implies $u_{i-1},u_{j+1} \in M_2$.
We observe that $\{u_{i+1},u_{j-1}\}$ is a virtual edge labeled by $\{u_i,u_j\}$.
In particular, $u_{i+1},u_{j-1}$ are isolated in $M_2$, and so, $u_{i+2},u_{j-1} \notin M_2$.
It implies, since $u_{i+2},u_{j-1} \in N_G(M_2)$ and $M_2$ is a module, $P'= (u_1,u_2,\ldots,u_{i-1},u_{i+2},\ldots,u_{j-1},u_{j+1},\ldots,u_{2\ell})$ is  shorter augmenting path than $P$, thereby contradicting the minimality of $P$.
\end{itemize}
Overall the first case implies, as claimed, $j=i+1$, whereas all other cases lead to a contradiction.
Therefore, $j=i+1$.
\end{proofclaim}

Finally, our last preparatory claim is that $P$ can cross the module $M_1$ in at most one edge.

\begin{claim}\label{claim:small-cross}
$|E(P) \cap F_1| \leq 1$.
\end{claim}

\begin{proofclaim}
Suppose by contradiction there exist $\{u_i,u_{i+1}\}, \{u_j,u_{j+1}\} \in F_1 \cap E(P)$, for some $i < j$.
Since $M_1$ induces a matching, $u_{i-1},u_{j-1} \notin M_1$.
There are three cases.
\begin{itemize}
\item Case $i,j$ even.
Then, $P'=(u_1,\ldots,u_{i-1},u_j,u_{j+1},\ldots,u_{2\ell})$ is a shorter augmenting path than $P$, thereby contradicting the minimality of $P$.
\item Case $i,j$ odd.
Then, $P'=(u_1,\ldots,u_i,u_{j-1},u_j,\ldots,u_{2\ell})$ is a shorter augmenting path than $P$, thereby contradicting the minimality of $P$.
\item Case $i$ even, $j$ odd (Case $i$ odd, $j$ even is symmetrical to this one).
Then, \\$P'=(u_1,\ldots,u_{i-1},u_{j+1},\ldots,u_{2\ell})$ is a shorter augmenting path than $P$, thereby contradicting the minimality of $P$.
\end{itemize}
We note that in order to prove this result, we did not use the fact that $M_1$ is pendant.
\end{proofclaim}

Let $\{u_{i_0},u_{i_0+1}\}$ be the unique edge in $E(P) \cap F_1$.
Such edge must exist since otherwise, $P$ would also be an $F^*$-augmenting path.
In order to derive a contradiction, we are left to replace $\{u_{i_0},u_{i_0+1}\}$ with a virtual edge.
We prove next that it can be easily done if $i_0$ is odd, {\it i.e.}, $\{u_{i_0},u_{i_0+1}\} \notin F$.
Indeed, in such case we observe that $\{u_{i_0-1},u_{i_0+2}\}$ is the virtual edge that is labeled by $\{u_{i_0},u_{i_0+1}\}$.
Furthermore, $\{u_{i_0-1},u_{i_0+2}\} \in F^*$ since we confirmed the {\sc Split}.
Therefore, we will assume from now on that $i_0$ is even, {\it i.e.}, $\{u_{i_0},u_{i_0+1}\} \in F$.

We will need the following observation:

\begin{claim}\label{claim:cross-m2}
The vertices $u_{i_0-1},u_{i_0+2}$ are the only vertices in $M_2 \cap V(P)$.
\end{claim}

\begin{proofclaim}
Suppose for the sake of contradiction this is not the case.
By symmetry, we can assume the existence of an index $j < i_0 - 1$ such that $u_j \in M_2$.
Furthermore, $j$ is even since otherwise, $P'=(u_1,\ldots,u_j,u_{i_0},u_{i_0+1},\ldots,u_{2\ell})$ would be a shorter augmenting path than $P$, thereby contradicting the minimality of $P$.
For the same reason as above, we also have $u_{j+1} \notin M_2$.
However, since $u_{j+1} \in N_G(M_2)$ and $M_2$ is a module, it implies that $P'=(u_1,\ldots,u_j,u_{j+1},u_{i_0+2},\ldots,u_{2\ell})$ would be a shorter augmenting path than $P$, thereby contradicting the minimality of $P$.
\end{proofclaim}

There are three cases.
\begin{enumerate}
\item Case $u_{i_0-1},u_{i_0+2} \notin V(F_2^*)$ (left unmatched by $F_2^*$).
There exists a virtual edge $\{x_2,y_2\}$ that is labeled by $\{u_{i_0},u_{i_0+1}\}$ (otherwise, the second phase could have continued with $u_{i_0-1},u_{i_0+2}$ and $\{u_{i_0},u_{i_0+1}\}$).
The two of $x_2,y_2$ cannot be matched together in $F^*$ since we have $\{u_{i_0},u_{i_0+1}\} \in F$. 
Nevertheless, since $x_2,y_2$ are adjacent in the subdivision $G^*({\cal P}^*)$, at least one of the two vertices, say $x_2$, is matched by $F^*$.
There are two subcases.
\begin{enumerate}
\item Subcase $y_2$ is exposed.
Let $\{w,x_2\} \in F^*$.
Since $w \neq x_2$ we have $w \notin M_2$.
Then, $P^*=(u_1,u_2,\ldots,u_{i_0-1},w,x_2,y_2)$ is an $F^*$-augmenting path, thereby contradicting the maximality of $F^*$.
\item Subcase $y_2$ is matched.
Let $\{w,x_2\},\{w',y_2\} \in F^*$.
As before, $w,w' \notin M_2$.
Then, $P^*=(u_1,u_2,\ldots,u_{i_0-1},w,x_2,y_2,w',u_{i_0+2},\ldots,u_{2\ell})$ is an $F^*$-augmenting path, thereby contradicting the maximality of $F^*$.
\end{enumerate}
\item Case $\{u_{i_0-1},u_{i_0+2}\} \in F_2^*$.
We have $\{u_{i_0-1},u_{i_0+2}\} \notin F^*$.
Hence, we have that $P^*= (u_1,u_2,\ldots,u_{i_0-1},u_{i_0+2},\ldots,u_{2\ell})$ is an $F^*$-augmenting path, thereby contradicting the maximality of $F^*$.
\item Case $\{u_{i_0-1},w\} \in F_2^*$ for some $w \neq u_{i_0+2}$.
There are two subcases.
\begin{enumerate}
\item Subcase $w$ is exposed.
Then, $P^*=(u_1,u_2,\ldots,u_{i_0-1},w)$ is an $F^*$-augmenting path, thereby contradicting the maximality of $F^*$.
\item Subcase $w$ is matched.
Let $\{w,w'\} \in F^*$.
As before, $w' \notin M_2$.
Then, \\$P^*=(u_1,u_2,\ldots,u_{i_0-1},w,w',u_{i_0+2},\ldots,u_{2\ell})$ is an $F^*$-augmenting path, thereby contradicting the maximality of $F^*$.
\end{enumerate}
\end{enumerate}
Summarizing, by the contrapositive we get $F^*$ maximum for $G^*({\cal P}^*)$ $\Longrightarrow$ $F$ maximum for $G$.
\end{proof}

\paragraph{Complexity.}
The complexity of this reduction rule is essentially dominated by {\sc Match} and {\sc Split} operations.
Therefore, the total running time is an ${\cal O}(\delta_1)$, that is in ${\cal O}(\Delta m(G') - \Delta m(G^*))$.

\subsection{Main result}\label{sec:main}

Our framework consists in applying any reduction rule presented in this section until it can no more be done.
Then, we rely on the following result:

\begin{theorem}[~\cite{CDP18}]\label{thm:main-soda}
For every $\langle G', {\cal P}, {\cal F} \rangle$, we can solve {\sc Module Matching} in ${\cal O}(\Delta{\cal \mu} \cdot p^4)$-time.
\end{theorem}

We are now ready to state our main result in this paper.

\begin{theorem}\label{thm:main}
Let $G=(V,E)$ be a graph.
Suppose that, for every prime subgraph $H'$ in the modular decomposition of $G$, its pruned subgraph has order at most $k$.
Then, we can solve {\sc Maximum Matching} for $G$ in ${\cal O}(k^4\cdot n + m\log n)$-time.
\end{theorem}

\begin{proof}
By Lemma~\ref{lem:incr}, it suffices to solve {\sc Module Matching} for any $\langle H', {\cal P}, {\cal F} \rangle$, with $H'$ in the modular decomposition of $G$, in time ${\cal O}(p + \Delta m \cdot \log{p} + k^4\cdot \Delta{\cal \mu})$.
For that, we start computing the pruned subgraph $H^{pr}$ of $H$, and a corresponding pruning sequence.
By Proposition~\ref{prop:pruned-compute}, it can be done in ${\cal O}(p + \Delta m \cdot \log{p})$-time.
Then, we follow the pruning sequence and at each step, we apply the reduction rule that corresponds to the current one-vertex extension.
Doing so, we pass by various intermediate instances $\langle H^j, {\cal P}^j, {\cal F}^j \rangle$.
For any rule we apply, the pre-processing time for passing from $\langle H^j, {\cal P}^j, {\cal F}^j \rangle$ to the next instance $\langle H^{j+1}, {\cal P}^{j+1}, {\cal F}^{j+1} \rangle$ is an ${\cal O}(\Delta m(H^j) - \Delta m(H^{j+1}))$.
Similarly, the post-processing time for computing a solution for $\langle H^j, {\cal P}^j, {\cal F}^j \rangle$ from a solution for $\langle H^{j+1}, {\cal P}^{j+1}, {\cal F}^{j+1} \rangle$ is an ${\cal O}(\Delta m(H^j) - \Delta m(H^{j+1}))$.
Therefore, the total running time for applying all the reduction rules is an ${\cal O}(\Delta m)$.
Finally, we are left with solving {\sc Module Matching} on a reduced instance $\langle H^{pr}, {\cal P}^{pr}, {\cal F}^{pr} \rangle$.
We stress that if $H'$ is degenerate (complete or edgeless) then $H^{pr}$ is trivial, otherwise by the hypothesis $H^{pr}$ has order at most $k$.
As a result, by Theorem~\ref{thm:main-soda} we can solve {\sc Module Matching} on the reduced instance in ${\cal O}(\Delta{\cal \mu} \cdot k^4)$-time.
\end{proof}

\section{Applications}\label{sec:applications}

We conclude this paper presenting applications of our main result to some graph classes (Theorem~\ref{thm:main}).
Some interesting refinements of our framework are also presented in Section~\ref{sec:unicycle}.

\subsection{Graphs totally decomposable by the pruned modular decomposition}

Cographs are exactly the graphs that are totally decomposable by modular decomposition~\cite{CPS85}.
We show that three distinct generalizations of cographs in the literature are totally decomposable by the pruned modular decomposition.

\paragraph{Distance-hereditary graphs.}
A graph $G=(V,E)$ is distance-hereditary if it can be reduced to a singleton by pruning sequentially the pendant vertices and twin vertices~\cite{BaM86}.
Conversely, $G$ is co-distance hereditary if it is the complement of a distance-hereditary graph, {\it i.e.}, it can be reduced to a singleton by pruning sequentially the anti-pendant vertices and twin vertices.
In both cases, the corresponding pruning sequence can be computed in linear-time~\cite{DHP01,DGE08}.
Therefore, we can derive the following result from our framework:

\begin{proposition}\label{prop:dh}
We can solve {\sc Maximum Matching} in linear-time on graphs that can be modularly decomposed into distance-hereditary graphs and co-distance hereditary graphs.

In particular, we can solve {\sc Maximum Matching} in linear-time on distance-hereditary graphs and co-distance hereditary graphs.
\end{proposition}

We stress that even for distance-hereditary graphs, we may need to use the reduction rule of Section~\ref{sec:pendant} for pendant modules.
Indeed, as we follow the pruning sequence, we may encounter twin vertices and merge them into a single module.
Hence, even in the simpler case of distance-hereditary graphs, we need to handle with modules instead of just handling with vertices.
In the same way, even for co-distance hereditary graphs, we may need to use the reduction rule of Section~\ref{sec:anti-pendant} for anti-pendant modules.

\medskip
Trees are a special subclass of distance-hereditary graphs.
We say that a graph has {\em modular treewidth} at most $k$ if every prime quotient subgraph in its modular decomposition has treewidth at most $k$\footnote{Our definition is more restricted than the one in~\cite{PSS16} since they only impose the quotient subgraph $G'$ to have bounded treewidth.}.
In particular, graphs with modular treewidth at most one are exactly the graphs that can be modularly decomposed into trees.
We stress the following consequence of Proposition~\ref{prop:dh}:

\begin{corollary}
We can solve {\sc Maximum Matching} in linear-time on graphs with modular-treewidth at most one.
\end{corollary}

The case of graphs with modular treewidth $k \geq 2$ is left as an intriguing open question.

\paragraph{Tree-perfect graphs.}

Two graphs $G_1,G_2$ are $P_4$-isomorphic if there exists a bijection from $G_1$ to $G_2$ such that, for every induced $P_4$ in $G_1$, its image in $G_2$ also induces a $P_4$~\cite{Chv84}.
The notion of $P_4$-isomorphism plays an important role in the study of perfect graphs.
A graph is {\em tree-perfect} if it is $P_4$-isomorphic to a tree~\cite{BrL99}.
We prove the following result:

\begin{proposition}\label{prop:tree-perfect}
Tree-perfect graphs are totally decomposable by the pruned modular decomposition.

In particular, we can solve {\sc Maximum Matching} in linear-time on tree-perfect graphs.
\end{proposition}

Our proof is based on a deep structural characterization of tree-perfect graphs~\cite{BrL99}.
Before stating this characterization properly, we need to introduce a few additional graph classes.

\begin{figure}[h!]
\centering
\includegraphics[width=.55\textwidth]{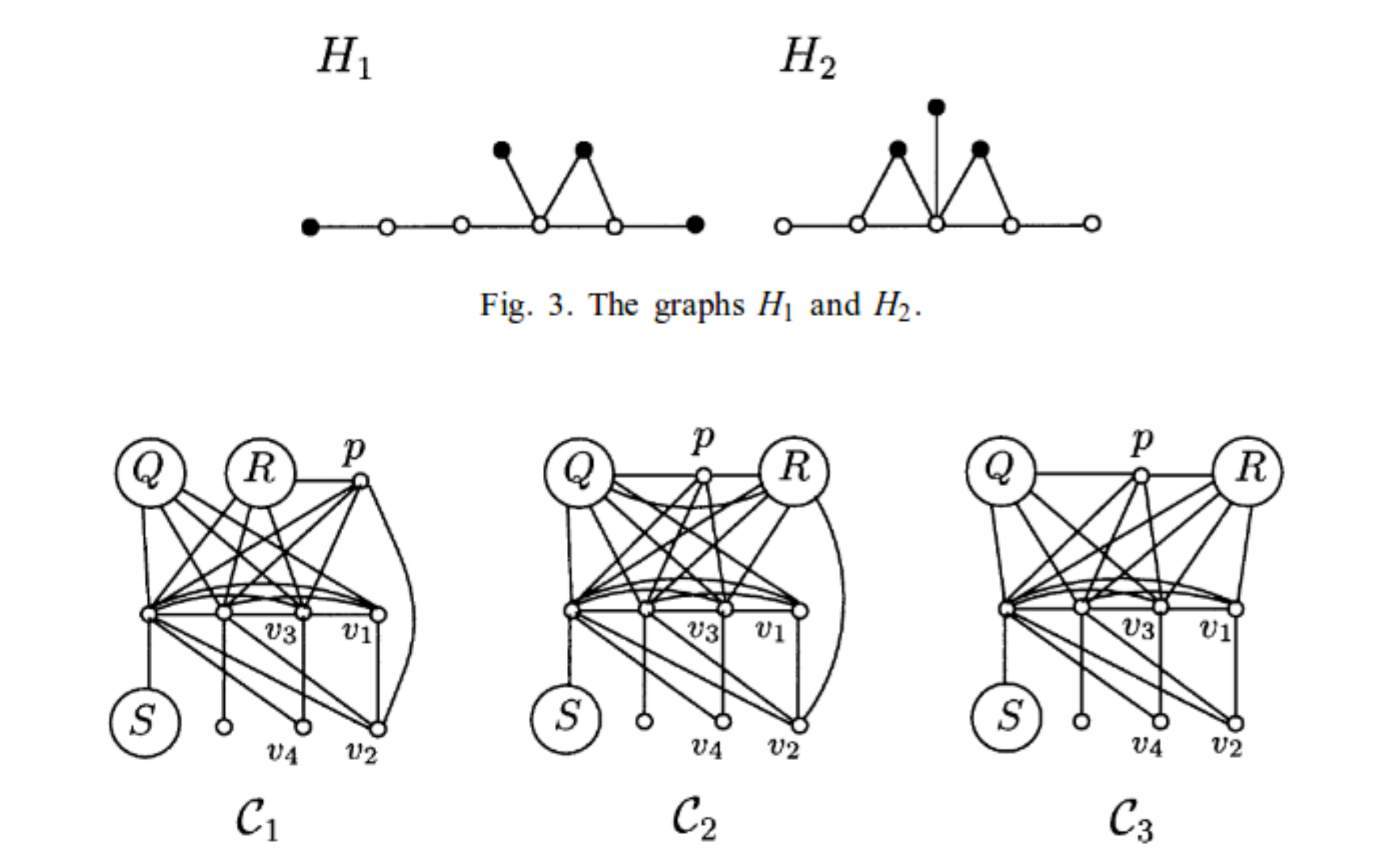}
\caption{Examples of tree-perfect graphs~\cite{BrL99}. The sets $Q,R,S$ represent modules substituting the vertices $q,r,v_n$.}
\label{fig:tree-perfect}
\end{figure}

Given a vertex-ordering $(v_1,v_2,\ldots,v_n)$ let $N_{<i}(v_i) = N(v_i) \cap \{v_1,v_2,\ldots,v_{i-1}\}$.
A graph is termed {\em elementary} if it admits a vertex-ordering $(v_1,v_2,\ldots,v_n)$ such that, for every $i$:
$$N_{<i}(v_i) = \begin{cases} \{v_1,v_2,\ldots,v_{i-2}\} \ \mbox{if} \ i \ \mbox{is odd} \\ \{v_{i-1}\} \ \mbox{otherwise.}\end{cases}$$
Note that such ordering as above is a pruning sequence by pendant and anti-pendant vertices.

The classes ${\cal C}_j, \ j = 1,2,3$ contain all the graphs that can be obtained from an elementary graph, with ordering $(v_1,v_2,\ldots,v_n)$, by adding the three new vertices $p,q,r$ and the following set of edges:
\begin{itemize}
\item (for all classes) $\{p,v_i\}, \{q,v_i\}, \{r,v_i\}$ for every $i > 1$ odd;
\item (only for ${\cal C}_1$) $\{v_1,q\}, \ \{p,r\} \ \mbox{and} \ \{v_2,p\}$; 
\item (only for ${\cal C}_2$) $\{p,q\}, \ \{p,r\}, \{q,r\}, \ \{v_1,q\} \ \mbox{and} \ \{v_2,r\}$;
\item (only for ${\cal C}_3$) $\{p,q\}, \ \{p,r\} \ \mbox{and} \ \{v_1,r\}$.
\end{itemize}
The graphs $H_1,H_2$ are illustrated in Fig.~\ref{fig:tree-perfect}

\medskip
Tree-perfect graphs are fully characterized in~\cite{BrL99}, and a linear-time recognition algorithm can be derived from this characterization.
We will only use a weaker form of this result:

\begin{theorem}[~\cite{BrL99}]\label{thm:tree-perfect}
A graph $G=(V,E)$ is a tree-perfect graph {\em only if} every nontrivial module induces a cograph and the quotient graph $G'$ is in one of the following classes or their complements: trees; elementary graphs; ${\cal C}_1 \cup {\cal C}_2 \cup {\cal C}_3$; $H_1$ or $H_2$.  
\end{theorem}

We can now apply Theorem~\ref{thm:tree-perfect} in order to conclude, as follows: 

\begin{proofof}{Proposition~\ref{prop:tree-perfect}}
Let $G=(V,E)$ be a tree-perfect graph.
By Theorem~\ref{thm:tree-perfect} every nontrivial module induces a cograph.
It implies that all the subgraphs in the modular decomposition of $G$, except maybe its quotient graph $G'$, are cographs, and so, totally decomposable by the modular decomposition.
We are left with proving that $G'$ is totally decomposable by the pruned modular decomposition.

The latter is immediate whenever $G$ is a tree, $H_1$, $H_2$ or a complement of one of these graphs.
Furthermore, we already observed that elementary graphs can be reduced to a singleton by pruning pendant and anti-pendant vertices sequentially.
Therefore elementary graphs and their complements are also totally decomposable by the pruned modular decomposition.

\begin{figure}[h!]
\centering
\includegraphics[width=.35\textwidth]{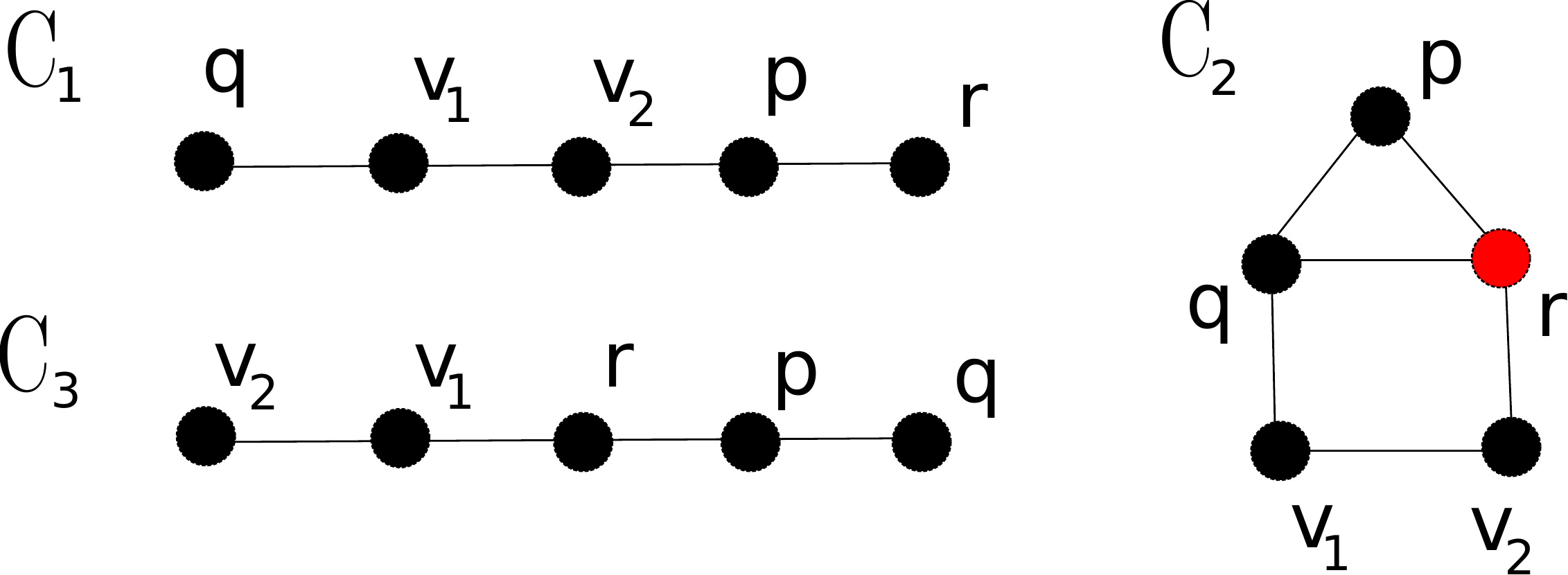}
\caption{Small tree-perfect graphs with $5$ vertices.}
\label{fig:tree-perfect-final}
\end{figure}

Finally, we prove that graphs in ${\cal C}_1 \cup {\cal C}_2 \cup {\cal C}_3$ are totally decomposable (this will prove the same for their complements).
Recall that every graph $G' \in {\cal C}_j, \ j = 1,2,3$ can be obtained from an elementary graph $H$ with ordering $(v_1,v_2,\ldots,v_n)$ by adding three new vertices $p,q,r$ and a set of specified edges. 
Furthermore, for every odd $i$, resp. for every even $i$, we have that $v_i$ is anti-pendant, resp. pendant, in $H \setminus \{v_{i+1},\ldots,v_n\}$.
Since $p,q,r$ are made adjacent to every $v_i$ for $i > 1$ odd, and nonadjacent to every $v_i$ for $i > 2$ even, this above property stays true in $G' \setminus \{v_{i+1},\ldots,v_n\}$.
As a result, we can remove the vertices $v_n,v_{n-1},\ldots,v_3$ sequentially.
We are left with studying the subgraph induced by $p,q,r,v_1,v_2$.
The latter subgraph is a path if $G' \in {\cal C}_1 \cup {\cal C}_3$, otherwise it is a house (cf. Fig.~\ref{fig:tree-perfect-final}).
In both cases, such subgraph can be totally decomposed by pruning pendant and anti-pendant vertices sequentially.
\end{proofof}

\paragraph{Other generalizations.}
Finally, the $c$-decomposition is a lesser-known generalization of the modular decomposition studied in~\cite{Lan01,Rao08}.
It was proved in~\cite{Lan01} that the graphs totally decomposable by the $c$-decomposition are exactly the graphs that can be reduced to a singleton by pruning pendant and universal vertices sequentially.

\begin{proposition}\label{prop:c-dec}
We can solve {\sc Maximum Matching} in linear-time on the graphs that are totally decomposable by the $c$-decomposition.
\end{proposition}

\subsection{The case of unicycles}\label{sec:unicycle}

We end up this section with a refinement of our framework for the special case of unicyclic quotient graphs ({\it a.k.a.}, graphs with exactly one cycle).

\begin{proposition}
We can solve {\sc Maximum Matching} in linear-time on the graphs that can be modularly decomposed into unicycles.
\end{proposition}

\begin{proof}
By Lemma~\ref{lem:incr}, it suffices to show that on every instance $\langle G', {\cal P}, {\cal F}\rangle$ such that $G'$ is a unicycle, we can solve {\sc Module Matching} in ${\cal O}(\Delta m)$-time.
Recall that $G'$ is a unicycle if it can be reduced to a cycle by pruning the pendant vertices sequentially.
Therefore, in order to prove the result, we only need to prove it when $G'$ is a cycle.

Given an edge $e = \{v_i,v_j\} \in E(G')$, our strategy consists in fixing the number $\mu_{i,j}$ of matched edges with one end in $M_i$ and the other end in $M_j$.
By~\cite[Lemma $5.1$]{CDP18}, we can always assume that the ends of these $\mu_{i,j}$ edges are the $\mu_{i,j}$ first vertices in a canonical ordering of $M_i$ (w.r.t. $F_i$), and in the same way, the $\mu_{i,j}$ first vertices in a canonical ordering of $M_j$ (w.r.t. $F_j$). 
We can remove these above vertices from $M_i,M_j$ and update the matchings $F_i,F_j$ accordingly.
Doing so, we can remove the edge $\{v_i,v_j\}$ from $G'$.
Then, since $G' \setminus e$ is a path, we can systematically apply the reduction rule for pendant modules (Section~\ref{sec:pendant}).
Overall, we test for all possible number of matched edges between $M_i$ and $M_j$ and we keep any one possibility that gives the largest matching.

In order to apply our strategy, we choose any edge $e$ such that $|M_i|$ is minimized.
Doing so, there can only be at most ${\cal O}(n_i) \leq {\cal O}(\Delta m/p)$ possibilities for $\mu_{i,j}$, where $p = |V(G')|$.
However, we are not done yet as we now need to test for every possibility in ${\cal O}(p)$-time.
A naive implementation of this test, using the reduction rule of Section~\ref{sec:pendant}, would run in ${\cal O}(\Delta m)$-time.
We propose a faster implementation that only computes the {\em cardinality} of the solution ({\it i.e.}, not the matching itself).
The latter is enough in order to compute the optimum value for $\mu_{i,j}$.
Then, once this value is fixed, we can run the naive implementation in order to compute a maximum-cardinality matching. 

\medskip
W.l.o.g., $i=1$, $j=p$.
For every $t$ let $n_t = |M_t|$.
Furthermore, let $\mu_t = |F_t|$.
Note that there are exactly $n_t - 2 \mu_t$ vertices in $M_t$ that are left exposed by $F_t$.
We also maintain a counter $\mu$ representing the cardinality of the current matching.
Initially $\mu = \mu_{i,j}$.
Then, we proceed as follows:
\begin{itemize}
\item We start removing the $\mu_{i,j}$ first vertices in a canonical ordering of $M_i$ w.r.t. $F_i$.
More precisely, we decrease $n_i$ by $\mu_{i,j}$.
If $\mu_{i,j} \leq n_i - 2 \mu_i$ then we only removed exposed vertices and there is nothing else to change.
Otherwise, we also need to decrease $\mu_i$ by exactly $\left\lceil \left( \mu_{i,j} - n_i + 2\mu_i \right) / 2 \right\rceil$.
We proceed similarly for $M_j$.

After that, we can remove $e$ from $G'$.
We have that $G' / e$ is isomorphic to the path $(v_1,v_2,\ldots,v_{p})$.
This first step takes constant-time.
\item Then, for every $1 \leq t < p$, we simulate the reduction rule of Section~\ref{sec:pendant} sequentially.
More precisely:
\begin{enumerate}
\item Let $k_t = \min\{n_t - 2\mu_t, n_{t+1}\}$ be the maximum number of exposed vertices in $M_t$ that can be matched with a vertex of $M_{t+1}$ in the first phase.
We decrease $n_t,n_{t+1}$ by $k_t$.
Furthermore, the size $\mu$ of the current matching is also increased by $k_t$.

If $k_t \leq n_{t+1} - 2 \mu_{t+1}$ then we only remove exposed vertices from $M_{t+1}$ and so, there is nothing else to be done.
Otherwise, we also need to decrease $\mu_{t+1}$ by exactly $\left\lceil \left( k_t - n_{t+1} + 2 \mu_{t+1} \right) / 2 \right\rceil$.

We fall in a degenerate case if $k_t = n_t$ or $k_t = n_{t+1}$.
In the former case, we do not modify the value of $\mu$, however in the latter case ($M_t$ is now an isolated module) we can increase this value by $\mu_t$.
For both degenerate cases, we continue directly to the next vertex $v_{t+1}$.
Otherwise, we go to Step 2.

\item Let $k'_t = \min\{ \left\lfloor ( n_{t+1} - 2 \mu_{t+1} )/2 \right\rfloor , \mu_t  \}$ be the number of virtual edges that we create during the second phase.
We increase $\mu_{t+1}$ by exactly $k'_{t}$.

\item Finally, in order to simulate the third phase, we claim that we only need to increase $\mu$ by exactly $\mu_t$.
Indeed, after a solution $F_t^*$ was obtained for $(v_{t+1},\ldots,v_{p})$ the reduction rule proceeds as follows.
Either we confirm a {\sc Split} operation, {\it i.e.}, we replace a virtual matched edge in $F_t^*$ by two edges between $M_t,M_{t+1}$; or we cancel the {\sc Split} operation, {\it i.e.}, we add an edge of $F_t$ in the current matching.
In both cases, the cardinality of the solution increases by one.
Then, all the edges of $F_t$ that were not used during the second phase are added to the current matching.
Overall, we have as claimed that the cardinality of the solution increases by exactly $\mu_t$. 
\end{enumerate}
\end{itemize}
The procedure ends for $t = p$.
In this situation, the quotient subgraph is reduced to a single node, and so, we only need to increase the current size $\mu$ of the matching by $\mu_{p}$.
Summarizing, since all the steps of this procedure take constant-time, the total running-time is an ${\cal O}(p)$.
\end{proof}

\section{Open problems}\label{sec:ccl}

The pruned modular decomposition happens to be an interesting add up in the study of {\sc Maximum Matching} algorithms.
An exhaustive study of its other algorithmic applications remains to be done.
Moreover, another interesting question is to characterize the graphs that are totally decomposable by this new decomposition.

We note that our pruning process can be seen as a repeated update of the modular decomposition of a graph after some specified modules (pendant, anti-pendant) are removed.
However, we can only detect a restricted family of these new modules (universal, isolated, twins).
A fully dynamic modular decomposition algorithm could be helpful in order to further refine our framework.

Finally, in a companion paper~\cite{DuP18+}, we propose another approach for {\sc Maximum Matching} that is based on split decomposition, and that partly overlaps the cases seen in this paper.
The combination of both framework looks like a challenging task.

\bibliographystyle{abbrv}
\small
\bibliography{bibliography-modular}

\end{document}